\documentclass[11pt,a4paper,pdftex]{scrartcl}
\usepackage[utf8]{inputenc}  
\usepackage[T1]{fontenc}
\usepackage[english]{babel}
\usepackage{amsmath}
\usepackage{amsthm}
\usepackage{amsfonts}
\usepackage{amssymb}
\usepackage[pdftex]{graphicx}
\usepackage{booktabs}
\usepackage{url}
\usepackage{color}
\usepackage{enumerate}
\usepackage{dsfont}
\usepackage{natbib}

\newcommand{\R}{\mathbb{R}}
\newcommand{\N}{\mathbb{N}}

\newcommand{\diff}{\operatorname{d}\!}

\newcommand{\C}{\mathcal{C}}
\newcommand{\Pc}{\mathcal{P}}
\newcommand{\M}{\mathcal{M}}

\newcommand{\scalprod}[2]{\left\langle #1, #2 \right\rangle}
\newcommand{\norm}[1]{\left\| #1 \right\|}
\newcommand{\betr}[1]{\left| #1 \right|}
\newcommand{\ind}{\mathds{1}}

\newcommand{\Con}[1]{\operatorname{Const}_{#1}}
\newcommand{\bN}{\mathbf{N}}

\newcommand{\tnzd}{t_{23}^{(k)}}
\newcommand{\tezd}{t_{23}^{(k+1)}}
\newcommand{\tnez}{t_{12}^{(k)}}
\newcommand{\tnezo}{t_{12}}

\newcommand{\texd}{t_{\text{ex}}}

\renewcommand{\L}{\mathcal{L}}
\newcommand{\Sd}{\mathcal{S}^{d-1}}

\ifdefined\hypersetup
\hypersetup{
	pdfkeywords={}, linkcolor=black, citecolor=black, filecolor=black, urlcolor=black,
}
\fi

\usepackage[hidelinks]{hyperref}

\usepackage[autostyle]{csquotes}

\renewcommand{\L}{\mathcal{L}}

\begin{document}
	
	\theoremstyle{plain}
	\newtheorem{lemma}{Lemma}[section]
	\newtheorem{corollary}[lemma]{Corollary}
	\newtheorem{proposition}[lemma]{Proposition}
	\newtheorem{definition}[lemma]{Definition}
	\theoremstyle{definition} 
	\newtheorem{example}[lemma]{Example}
	\newtheorem{remark}[lemma]{Remark}
	\newtheorem{theorem}{Theorem}
	\newtheorem*{theorem*}{Theorem}
	
	\title{A general improbability result for non-collision singularities as subsystems with at most four particles}
	\author{Manuel Quaschner\thanks{Institut für Mathematik, Friedrich Schiller Universität Jena, 07737 Jena, Germany}}
	\maketitle
	
	\begin{abstract}
			Non-collision singularities of the $n$-body problem are initial conditions for which no global solution exists and which, however, do not lead to a collision in the limit; i.e., the moment of inertia of the system diverges. 
		The question of whether the set of initial conditions leading to non-collision singularities in the $n$-body problem is improbable is open and is the first on Barry Simon's list of fifteen problems in mathematical physics \cite{simon1984fifteen} from the year 1984. So far, this question has only been answered positively in the case $n=4$. \\
		Using a final cluster decomposition into subsystems whose particles interact strongly with each other and exert only small forces on particles from different subsystems, we can prove the improbability of the set of all singular orbits with the following conditions:
		\begin{itemize}
			\item There may be arbitrarily many subsystems undergoing total collision at their center of mass.
			\item There may be subsystems exhibiting a non-collision singularity, i.e. the respective moment of inertia of the subsystem diverges. These subsystems consist of exactly four particles each and the asymptotic directions of these subsystems have to differ (which is only a mild condition).
		\end{itemize}
		This result includes all known statements and indicates a way how the general problem could be solved.
	\end{abstract}

	\section{Introduction}
	
	In his lectures in Stockholm from 1895, Painlevé conjectured the existence of so-called \enquote{non-collision singularities} in the $n$-body problem, i.e. solutions that become unbounded within finite time.
	In order to describe this problem and some of the currently known results about this question, we first fix some notation similar to that in \cite[Chapter 1]{quaschner2023non}.\\
	We formulate the $n$-body problem as a Hamiltonian system and define the energy function in terms of the particle positions $q$ and momenta $p$:
	\begin{align}
		H(q,p) := \sum\limits_{i=1}^{n}\frac{\norm{p_i}^2}{2 m_i} + V(q) \label{eq:Hamiltonfunktion}
	\end{align}
	with \textbf{potential}
	\begin{align*}
		V(q) := \sum_{1 \leq i < j \leq n} V_{i,j}(q_i-q_j).
	\end{align*}
	The \textbf{pair potentials} $V_{i,j}$ are assumed to be homogeneous of degree $-\alpha$ for some \linebreak $\alpha \in (0,2)$, i.e.
	\begin{align}
		V_{i,j}(x) = \frac{Z_{i,j}}{\norm{x}^{\alpha}}, \; x \in \R^{d} \setminus \{0\}, \label{def:PairPotentials}
	\end{align}
	for interaction constants $Z_{i,j} \in \R, \; i,j \in \{1,\ldots, n\} =: \bN$. The constants $m_i>0$, $i=1, \ldots, n,$ are the \textbf{masses} of the particles. 
	For the main result of this paper, we assume $Z_{i,j} := - m_i \cdot m_j$, although some of the results can be formulated for more general potentials (between particles of the same final cluster).
	\\
	We define the \textbf{position space}
	\begin{align*}
		M &:= \R^{nd} \setminus \Delta
		\intertext{where}
		\Delta &:= \{ (q_1, \ldots, q_n) \in \R^{nd} \mid \exists i, j \in \bN, \; i \neq j: \; q_i = q_j \}
	\end{align*}
	is the collision set.
	The \textbf{phase space} for the system is the cotangent bundle
	\begin{align*}
		P&:= T^{*}\left(\R^{nd} \setminus \Delta \right) \simeq \left(\R^{nd}_q \setminus \Delta\right) \times \R^{nd}_p,
	\end{align*}
	$P$ carries a natural symplectic form $\omega$ that can be expressed in coordinates as
	\begin{align*}
		\omega &= \sum_{i=1}^{n} \sum_{j=1}^{d} \diff q_{i,j} \wedge \diff p_{i,j},
	\end{align*}
	where we write $(q,p) \in P$ in the form
	\begin{align*}
		q&=(q_1, \ldots, q_n) \in \R^{nd}_q \simeq \left(\R^{d}\right)^{n}
		\intertext{with}
		q_i &= \left(q_{i,1}, \ldots, q_{i,d}\right) \in \R^{d} \text{ for } i \in \bN
	\end{align*}
	(and an analogous notation for the momenta in $p$).
	\\
	Using the symplectic form and the Hamilton function defined on $P$, we get a Hamiltonian vector field $X_{H}$ defined through:
	\begin{align*}
		\mathbf{i}_{X_H} \omega = \diff H.
	\end{align*}
	In coordinates we obtain the well-known Hamilton equations with the Euclidean gradient $\nabla$ on $\R^{d}$:
	\begin{align*}
		\dot{q}_i &= \frac{p_i}{m_i}, \\
		\dot{p}_i &= -\sum_{j \in \bN \setminus \{i\}} \nabla V_{i,j}(q_i-q_j) = \sum_{j \in \bN \setminus \{i\}} \alpha Z_{i,j} \frac{q_i-q_j}{\norm{q_i-q_j}^{2+\alpha}}.
	\end{align*}
	Our goal is to investigate the \textbf{flow} generated by the vector field $X_H$ that is denoted by $\Phi : D \rightarrow P$ and is defined on a maximal domain $D$ of the form
	\begin{align*}
		D := \left\{ (t,x) \in \R \times P \mid T^{-}(x) < t < T^{+}(x) \right\}
	\end{align*}
	where $T^{-}, T^{+} : P \rightarrow \R \cup \{\pm \infty\}$ are the \textbf{escape times} defining the maximal interval of existence for the solution of this differential equation with initial condition $x$, satisfying $T^{+}>0, \; T^{-}<0$.
	\begin{definition}[Singularities and non-collision singularities] \label{def:SingAndNCSing} \quad
		\begin{enumerate}[(i)]
			\item 	An initial condition $x=(q,p) \in P$ is called a \textbf{singularity}, if the escape time\footnote{We only consider the forward time direction and thus the upper escape time $T^{+}$, but one should note that the analogous results also hold for considerations in negative time direction up to $T^{-}$ due to reversibility of the system.} for this initial condition is finite, i.e. $T^{+}(x) < \infty$. We write
			\begin{align*}
				\operatorname{Sing} &:= \{ x \in P \mid T^{+}(x)< \infty \}
			\end{align*}
			for the set of singularities. Since all points on the same orbit have a finite escape time, we sometimes talk about \textbf{singular orbits} as well.
			\item An initial condition $x=(q,p) \in \operatorname{Sing}$ is called a \textbf{collision singularity}, if for the solution curve $(\tilde{q}(t), \tilde{p}(t)) := \Phi(t, x)$ we have
			\begin{align*}
				\lim\limits_{t \uparrow T^{+}(x)} \tilde{q}(t) \quad \text{exists in $\R^{nd}$.}
			\end{align*}
			\item An initial condition $x=(q,p) \in \operatorname{Sing}$ is called a \textbf{non-collision singularity}, if for the solution curve $(\tilde{q}(t), \tilde{p}(t)) := \Phi(t, x)$ we have
			\begin{align*}
				\lim\limits_{t \uparrow T^{+}(x)} \norm{\tilde{q}(t)} = \infty.
			\end{align*}
			We denote the set of non-collision singularities with $\operatorname{NColl}$. Note that by the theorem of von Zeipel (see Section~\ref{sec:ClassicalResults}) we have
			\begin{align*}
				\operatorname{Sing} &= \operatorname{NColl} \cup \operatorname{Coll},
			\end{align*}
			where $\operatorname{Coll}$ denotes the set of initial conditions leading to a collision.
		\end{enumerate}
	\end{definition}
	Painlevé already proved that for $n \leq 3$ there are no non-collision singularities (the argument will be adapted to potentials with general $\alpha \in (0,2)$ and arbitrary spatial dimensions $d \geq 2$, see Section~\ref{sec:ClassicalResults}). But he conjectured that there could be such orbits if there are four or more particles. The conjecture was finally proven to be true by Xia, who constructed in 1992 an example with five bodies in dimension $d=3$ using techniques from \cite{mcgehee1974triple, mather1975solutions}. For a qualitative description of this behavior of solutions consisting of three clusters with one messenger cluster oscillating faster and faster between the two outer clusters we refer to \cite{saari1995off}. By now there are some other constructed non-collision singularities, even with $d=2$ and $n=4$, compare \cite{xue2020non,gerver2022new,gerver1991existence}.
	\\
	Since the question of the existence of non-collision singularities has been clarified, another question becomes important: Which properties can one derive for the set of non-collision singularities? Is this an \enquote{exceptional} behavior and is this set in some sense \enquote{small}? One way to answer this question would be a statement about the Lebesgue measure of this set in phase space or, more precisely, if this set is \textbf{improbable}, i.e. contained in a set of Lebesgue measure zero. Results of this type were first derived by Saari in \cite{saari1977global} and with the help of the so-called \textbf{Poincaré surface method} in \cite{fleischer2018improbability,knauf2018asymptotic}. All of these results show that non-collision singularities are improbable in the case of $n=4$ bodies and for potentials as in \eqref{def:PairPotentials} with $\alpha \in (0, 2)$ (some of them allow even more general potentials).  The importance of the question in this case can be seen by noting that the problem of the improbability of non-collision orbits is even on the list of Barry Simon's fifteen problems in mathematical physics \cite{simon1984fifteen} from the year 1984. 
	\\
	A partial result for more than four bodies was derived in \cite{quaschner2025improbability}
	by splitting the system according to the final cluster decomposition into subclusters of particles that interact with each other directly or indirectly arbitrarily close to the escape time, compare Definition~\ref{def:TempAndFinalClusterDecomposition} for a formal definition.
	However, that paper only allowed one cluster to exhibit a non-collision singularity and some restriction on the number of colliding particles in a cluster was necessary. In this paper, we can handle arbitrarily many subsystems of four particles in the final cluster decomposition having a non-collision singularity, as long as a mild condition on the asymptotic direction\footnote{Compare Proposition~\ref{prop:ExistenceAsymptoticDirection}} is satisfied. Additionally, we can handle collisions with arbitrarily many particles, as we can use a way to control the energy exchange between the subsystems that was derived in \cite{sperling1972collision}, compare also \cite{duignan2026blowup}
	with a notation similar to this paper.
	\\
	Now we state the main theorem of this paper
	\begin{theorem}[Improbability result for singularities with only subsystems of four particles having non-collision singularities] \label{thm:ImpResSubsystems} \quad \\
		Consider the $n$-body problem with an arbitrary $n \in \N$ given through \eqref{eq:Hamiltonfunktion} with $Z_{i,j} = - m_i \cdot m_j$. Consider an arbitrary final cluster decomposition $\C = \{ U_1, \ldots, U_k, B_1, \ldots, B_l\}$ with $k,l \in \N_0$ and $\left|U_i\right|= 4$ for $i=1, \ldots, k$ into unbounded and bounded subsystems. Then the set
		\allowdisplaybreaks
		\begin{align*}
			\Bigg\{x \in \operatorname{Sing} &\mid \C_{\text{fin}}(x) = \C, \forall i \in \{1, \ldots, k\}: \limsup_{t \uparrow T^{+}(x)} J_{U_i}^{I}(t, x) = \infty, \\ & \qquad \forall j \in \{1, \ldots, l\}: \limsup_{t \uparrow T^{+}(x)} J_{B_j}^{I}(t, x) < \infty, \\
			& \qquad \forall i < j \in \{1, \ldots, k\}: AD_{U_i} \neq \pm AD_{U_j} \Bigg\}
		\end{align*}
		is improbable, i.e. contained in a set of measure zero. 
	\end{theorem}
	
	As we can take the finite union over all possible final cluster decompositions of this type, all singularities of this type are improbable. Note that this result was already proven for $k=0$ in \cite{fleischer2019improbabilityCollisions}.
	
	We believe that results of this type are relevant to the general question of improbability. Using the final cluster decomposition from Definition~\ref{def:TempAndFinalClusterDecomposition} seems to be reasonable, as the particles from different subsystems have only a small influence (for collisions this was done in \cite{duignan2026blowup}
	and for unbounded subsystems, i.e. those experiencing a non-collision singularity, we can control this influence by techniques developed in Section~\ref{sec:UnboundedSubsystem} at least for four particles). Thus, it is necessary to understand the behavior of \enquote{simple} non-collision singularities with $a \in \N$ particles first and then combine these statements to the total statement (as long as one can control the perturbation by small forces). 
	\\
	The outline of this paper is as follows: 
	\\ In Section~\ref{chap:Preliminaries} we recall some useful notation on cluster coordinates and the decomposition of some quantities as angular momentum, moment of inertia or kinetic energy with respect to this decomposition. We then define the final cluster decomposition for general singularities and the asymptotic direction for unbounded subsystems with four particles. Afterwards we give a short recap of the Poincaré surface method. This will be our main tool in the proofs of the following improbability results. Finally we recap some classical results on the qualitative behavior of non-collision singularities, namely generalized versions of the theorems of Painlevé and von Zeipel.
	\\
	In Section~\ref{sec:CollidingSubsystems} we briefly recall some facts about subsystems suffering a total collision and how they can influence other subsystems. \\
	Section~\ref{sec:UnboundedSubsystem} is devoted to the unbounded subsystems. First, we can estimate the forces between the subsystems using a double-cone argument (for which we need the assumptions on the asymptotic directions). This allows us to apply the general theory from \cite{quaschner2025improbability}. If we only have one unbounded subsystem, these estimates allow us to define a hypersurface for this subsystem. If we have more than one unbounded subsystem, we need to control the motion even at the times of an almost-collision. This will be done in the second part of this section.
	\\
	Finally, in Section~\ref{sec:PoincareSurfaces} we define the Poincaré surfaces for our improbability result (using some exhaustion arguments) such that the desired singularities are transition points for this sequence. Then we prove that these hypersurfaces have finite and decreasing volume, which implies the theorem.

	\paragraph{Some Notation}
	\label{par:SomeNotations}
	\quad \\
	We use the notation
	\begin{align*}
		\operatorname{Const}
	\end{align*}
	to describe some generic constants. These constants may vary from one occurrence to another. Usually, they depend only on general system quantities as the masses, particle numbers, the dimension $d$ or interaction constants of the Hamilton function considered in the described context. If there is dependence on an additional constant $\epsilon$ arising from a further context, we write
	\begin{align*}
		\operatorname{Const}(\epsilon).
	\end{align*}
	Sometimes we omit this notation if we explicitly mention the dependencies of the constants in the surrounding text and this increases readability (e.g. if the dependence is on too many constants). \\
	If we use a constant with a fixed value, we indicate this by a subscript, e.g.
	\begin{align*}
		\Con{1}
	\end{align*}
	Within the given context, $\Con{1}$ will always denote the same value.
	\\
	For error terms we use an adapted version of the Landau $\mathcal{O}$-notation. For a real quantity $z$ we denote with
	\begin{align*}
		\mathcal{O}(z)
	\end{align*}
	some term $g$ with $\norm{g} \leq z$. Note that whether $g$ is a vector or a scalar will always become clear from the context. The norm we choose is usually the Euclidean norm, but most of the time this does not matter, as we can compensate the use of different norms by constants. As we want the estimate $\norm{g} \leq z$ and not only $\norm{g} \leq \operatorname{Const} z$, we will sometimes have generic constants in the $\mathcal{O}$-notation as well, e.g. $\mathcal{O}(\operatorname{Const} \cdot z)$. As with constants $\operatorname{Const}$, the term represented by $\mathcal{O}(z)$ can vary at any occurrence even within the same set of equations.

	\section{Preliminaries}
	\label{chap:Preliminaries}
	
	In this chapter, we introduce some notation and well-known results from the literature. The notation introduced first, especially the cluster coordinates and, based on these, the Jacobi coordinates, is important for understanding the subsequent chapters. Then we define the final cluster decomposition. For final clusters that correspond to non-collision singularities with four particles, we also define the asymptotic direction.
	Afterwards we recapitulate the method of Poincaré surfaces as introduced in \cite{fleischer2019improbability}. This is the main technique used in all improbability proofs in this paper. As a last point, we revisit some classical results of Painlevé and von Zeipel in a more modern fashion based on ideas of Sperling and developed in \cite{duignan2026blowup}. 
	Parts of this section are similar to the corresponding Section~2 in \cite{quaschner2023non}, compare also the references therein. 
	
	\subsection[Cluster Decomposition and Suitable Coordinates]{Cluster Decomposition and Suitable Coordinates} \quad
	
	We recap some notions on the set $\mathcal{P}(\bN)$ of set partitions of the finite set $\bN = \{1, \ldots, n \}$ and the lattice structure on this set. These notions are taken from \cite[Definition 2.2, p.5]{knauf2018asymptotic}, compare as well \cite [Chapter 12.6, p.311 et seqq]{knauf2018mathematical}.
	\begin{definition}[Set partition] \label{def:SetPartition}\quad
		\begin{enumerate}[(i)]
			\item A \textbf{set partition} or \textbf{cluster decomposition} of $\bN$ is a set \linebreak $\mathcal{C}:= \{C_1, \ldots, C_k\}$ of blocks or clusters $C_i \subseteq \bN, C_i \neq \emptyset, \; \forall i=1, \ldots, k$ with $C_i \cap C_j = \emptyset$ for $i \neq j$
			and $\bigcup\limits_{i=1}^{k} C_i = \bN$.
			\item The set $\mathcal{P}(\bN) := \{ \mathcal{C} \mid \mathcal{C} \text{ is a set partition of } \bN \}$ becomes a lattice with the partial order defined by the refinement relation\footnote{This text follows the notation of \cite{aigner2012combinatorial} for the order relation $\preccurlyeq$ between cluster decompositions. There is, however, also the convention of defining the comparison relation for set partitions the other way around. In this case, the definitions of join and meet have to be adapted accordingly. As a mnemonic one can say that a partition is coarser if it contains larger clusters (but less clusters than the finer partition).}
			\begin{align*}
				\mathcal{C} = \{C_1, \ldots, C_k \} &\preccurlyeq \{D_1, \ldots, C_l \} = \mathcal{D}
				\intertext{if there is a surjection $\pi: \{1, \ldots, k\} \rightarrow \{1, \ldots, l\}$ with}
				C_{i} &\subseteq D_{\pi(i)} \; \forall i \in \{1, \ldots, k\}.
			\end{align*}
			In this case, $\mathcal{C}$ is called \textbf{finer} than $\mathcal{D}$ and $\mathcal{D}$ is called \textbf{coarser} than $\mathcal{C}$.
			\item For two set partitions $\mathcal{C}$, $\mathcal{D}$ the set partition $\mathcal{C} \vee \mathcal{D}$ is defined to be the finest set partition that is coarser than both $\mathcal{C}$ and $\mathcal{D}$. It is called the \textbf{join} of $\mathcal{C}$ and $\mathcal{D}$. Analogously we define the \textbf{meet} $\mathcal{C} \wedge \mathcal{D}$ as the coarsest set partition that is finer than both $\mathcal{C}$ and $\mathcal{D}$.
			\item The finest set partition is $\mathcal{C}_{\min} := \{ \{1\}, \ldots, \{n\} \}$ and the coarsest set partition is $\mathcal{C}_{\max} = \{\bN\}$.
			\item Each set partition $\C \in \Pc(\bN)$ defines an equivalence relation on $\bN$ denoted $\sim_{\C}$ which is defined such that $\C$ is the set of equivalence classes.
		\end{enumerate}
	\end{definition}
	
	In our application to $n$-body problems, a cluster decomposition represents a division of the particles, for example we could group together particles that are close to each other in some sense. We need some useful coordinates, which give us a better control of the motion for those groups.
	\begin{definition}[Cluster coordinates] \label{def:ClusterCoordinates} \quad
		\begin{enumerate}[(i)]
			\item For a nonempty subset $C \subset \bN$ we define the \textbf{cluster mass} $m_C$ as
			\begin{align*}
				m_C &:= \sum_{i \in C} m_i
			\end{align*}
			\item For a cluster $B \subseteq \bN$ the \textbf{center of mass} $q_{B}$ and the \textbf{total momentum} $p_{B}$ are defined as
			\begin{align*}
				q_{B} &:= \frac{1}{m_{B}} \sum_{i\in B} m_i q_i, \\
				p_{B} &:=\sum_{i \in B} p_i.
			\end{align*}
			\item For a nonempty subset $C' \subsetneq B \subseteq \bN$ we define the \textbf{cluster barycenter} (or \textbf{cluster center of mass}) $q_{C'}$ and the \textbf{cluster momentum} $p_{C'}$ \textbf{with respect to the center of mass of the particles in $B$} as
			\begin{align*}
				q_{C'} &:= \frac{1}{m_{C'}}\sum_{i \in C'} m_i q_i - q_{B}\\
				p_{C'} &:=  \sum_{i \in C'} p_i - \frac{m_{C'}}{m_{B}} p_{B}.
			\end{align*}
			\label{def:ClusterCoordinates_external_Coordinates}
			Usually, the cluster $B$ is clear from the context and not stressed explicitly (e.g. if we consider the motion of a subsystem $B$ of the whole $n$-body system). If no explicit $B$ is given, one can take $B=\bN$.
			\item For two clusters $C_i, C_k \subseteq \bN$ the \textbf{internal distance} is
			\begin{align*}
				q_{C_i, C_k}^{I} &:= q_{C_i} - q_{C_k}
				\intertext{and the internal momentum is}
				p_{C_i, C_k}^{I} &:= \left(\frac{p_{C_i}}{m_{C_i}} - \frac{p_{C_k}}{m_{C_k}}\right) m_{C_i, C_k}^{I}
			\end{align*}
			with the \textbf{internal mass} (or \textbf{reduced mass}) $m_{C_i, C_k}^{I} = \frac{m_{C_i} m_{C_k}}{m_{C_i}+ m_{C_k}}$. \\
			If $D = \{j,m\}$ is a cluster with only two particles, we use the shorthand notation $q_{D}^{I} = q_{\{j\}, \{m\}}^{I}$.
		\end{enumerate}
	\end{definition}
	
	\begin{remark} \label{rem:ClusterCoordinates}\quad \\
		Most of Definition \ref{def:ClusterCoordinates} is common in the literature. But the additional subtraction of the total center of mass and total momentum in \eqref{def:ClusterCoordinates_external_Coordinates} relative to some subsystem might be surprising. Usually in $n$-body systems one works with the center of mass frame, as the total momentum is preserved. But as already pointed out in the introduction, we are sometimes interested in separated subsystems (or other cases where for some reason the center of mass is not preserved). Then it is more convenient to study these quantities relative to the center of mass of the corresponding subsystem. If we are, however, in the situation to work in the center of mass frame, both definitions coincide. \\
		One should note that if we consider internal distances of two clusters, the center of mass of the corresponding subsystem cancels out and we can work with the usual definition again. This will be useful in the following Definition \ref{def:JacobiCoordinates}, as the Jacobi coordinates are usually defined using the other definition of cluster coordinates. \hfill $\diamond$
	\end{remark}
	
	Next, we define several system quantities that are used frequently in this text. Some of these quantities will get adapted to clusters or cluster decompositions as well. The following definition is similar to \cite [p. 7ff and p.22]{knauf2018asymptotic}.
	
	\begin{definition}[System quantities] \label{def:System_quantities}\quad
		\begin{enumerate}[(i)]
			\item The \textbf{angular momentum} $L: T^{*}M \rightarrow \R^d \wedge \R^d$ is defined as
			\begin{align*}
				L(p,q) &:= \sum_{i=1}^{n} q_i \wedge p_i.
			\end{align*}
			\item Let $C \subseteq \bN$ be some cluster. Then the \textbf{(external) angular momentum} of cluster $C$ is defined as
			\begin{align*}
				L_C(p,q) &:= q_{C} \wedge p_{C}.
			\end{align*}
			Note that this quantity depends on some subsystem $B \subseteq C$ via the definition of $q_C$ and $p_C$.
			The \textbf{internal angular momentum} of cluster $C$ is defined as
			\begin{align*}
				L_{C}^{I}(p,q) &:= \sum_{j \in C} \left(q_{\{j\}} - q_C \right) \wedge \left(p_{\{j\}}- \frac{m_j}{m_C}p_C\right).
			\end{align*}
			Note that we take here for $q_{\{j\}}$ and $q_C$ (respectively $p_{\{j\}}$ and $p_C$) the same reference subsystem $B$, hence the influence of the reference subsystem cancels.
			\item Let $C_1, C_2 \subseteq \bN$ be clusters. The \textbf{internal angular momentum} between clusters $C_1$ and $C_2$ is defined as
			\begin{align*}
				L_{C_1, C_2}^{I}(p,q) &:= q_{C_1,C_2}^{I} \wedge p_{C_1, C_2}^{I}.
			\end{align*}
			\item The \textbf{moment of inertia} $J: T^{*}M \rightarrow [0, \infty)$ is defined as
			\begin{align*}
				J(p,q) \equiv J(q) := \frac{\norm{q}_{\M}^2}{2} = \sum_{i=1}^{n} \frac{m_i}{2} \norm{q_i}^2.
			\end{align*}
			Here, $\M$ is the mass matrix given as a (block) diagonal matrix
			\begin{align*}
				\M &= \operatorname{diag}\left(m_1 \ind_{d}, \ldots, m_n \ind_{d}\right),
			\end{align*}
			where $\ind_{d}$ denotes the $d\times d$ identity matrix.
			\item Let $C \subseteq \bN$ be some cluster. Then the \textbf{(external) moment of inertia} of the cluster $C$ (depending on some reference subsystem) is defined as
			\begin{align*}
				J_C(p,q) &:= \frac{m_C}{2} \norm{q_C}^2
			\end{align*}
			and the \textbf{internal moment of inertia} of the cluster $C$ is defined as
			\begin{align*}
				J_C^{I}(p,q) := \sum_{i \in C} \frac{m_i}{2} \norm{q_{\{i\}} - q_C}^2.
			\end{align*}
			\item Let $C_1, C_2 \subseteq \bN$ be clusters. The \textbf{internal moment of inertia} of the clusters $C_1$ and $C_2$ is defined as
			\begin{align*}
				J_{C_1, C_2}^{I}(p,q) &:= \frac{m_{C_1,C_2}^{I}}{2} \norm{q_{C_1,C_2}^{I}}^2.
			\end{align*}
			\item The \textbf{kinetic energy} $K: T^{*}M \rightarrow [0, \infty)$ is defined as
			\begin{align*}
				K(p,q) \equiv K(p) := \frac{\norm{p}_{\M^{-1}}^2}{2} = \sum_{i=1}^{n} \frac{\norm{p_i}^2}{2 m_i}.
			\end{align*}
			\item  Let $C \subseteq \bN$ be some cluster. Then the \textbf{(external) kinetic energy} of cluster $C$ (depending on some reference subsystem) is defined as
			\begin{align*}
				K_C(p,q) &:= \frac{\norm{p_C}^2}{2 m_C}
			\end{align*}
			and the \textbf{internal kinetic energy} of cluster $C$ is defined as
			\begin{align*}
				K_C^{I}(p,q) &:= \sum_{i \in C} \frac{\norm{p_{\{i\}}- \frac{m_i}{m_C} p_C}^2}{2 m_i}.
			\end{align*}
			\item Let $C_1, C_2 \subseteq \bN$ be clusters. The \textbf{internal kinetic energy} between clusters $C_1$ and $C_2$ is defined as
			\begin{align*}
				K_{C_1, C_2}^{I}(p,q) &:= \frac{\norm{p_{C_1,C_2}^{I}}^2}{2 m_{C_1,C_2}^{I}}.
			\end{align*}
			\item Let $C \subseteq \bN$ be some cluster. Then the \textbf{internal energy of the cluster} is defined as
			\begin{align*}
				H_{C}^{I}(p,q) = K_{C}^{I}(p,q) + \sum_{\substack{i,j \in C \\ i < j}} V_{i,j}(q_i-q_j).
			\end{align*}
		\end{enumerate}
	\end{definition}
	
	\begin{remark} \quad
		\begin{enumerate}[(i)]
			\item Usually, we regard the arguments of these quantities as implicit and clear from the context.
			\item When using \enquote{internal} quantities, one has to be careful whether this refers to one or two clusters. 
			\item \label{rem:IntCoordinates} The internal coordinates of one cluster $C$ have the interpretation of ignoring all particles not in $C$ and then considering the system in the center of mass frame. In this setting the internal quantities are just the total quantities of the so transformed system of particles in $C$. If we add the respective external quantity, we get the total quantity for the subsystem (without considering the subsystem in its center of mass frame). 
			\item The internal coordinates between two clusters $C_1, C_2$ are used for describing the difference of the respective cluster barycenters. In the case of $C_1 = \{i\}, C_2 = \{j\}$ for $i, j \in \bN$ we have as one would expect for any quantity $S \in \{L, J, K\}$
			\begin{flalign*}
				S_{\{i,j\}}^{I} &= S_{\{i\}, \{j\}}^{I}. & \diamond
			\end{flalign*}
		\end{enumerate}
	\end{remark}
	
	There is a more structured approach to cluster coordinates for a fixed cluster decomposition $\mathcal{C}$ in \cite[p.7]{knauf2018asymptotic}, where internal and external coordinates can be viewed as orthogonal projections in phase space. This yields a very simple proof of the following lemma:
	
	\begin{lemma}[Split of system quantities for cluster decomposition] \label{lem:SplitSystemQuantitiesClusterDecomposition} \quad \\
		Let $\mathcal{C} \in \Pc(\bN)$ be a cluster decomposition. Define for any $S \in \{L, J, K\}$ the respective external/internal cluster quantity with respect to $\bN$ as
		\begin{align*}
			S_{\C} &:= \sum_{C \in \C} S_C, \\
			S_{\C}^{I} &:= \sum_{C \in \C} S_C^{I}.
		\end{align*}
		Then we have the decomposition of the total quantity $S$ with respect to the cluster decomposition $\C$, namely
		\begin{align*}
			S = S_{\C} + S_{\C}^{I} + S_{\bN}.
		\end{align*}
	\end{lemma}

	We now derive another coordinate system that forms a complete set of coordinates and satisfies additional desirable properties. The idea is rather simple: We subsequently combine two particles/clusters to a larger cluster and take the internal distance of the single particles/clusters as a coordinate. After $n-1$ steps we add the total center of mass as last coordinate and have all position coordinates. Combined with the corresponding momenta, we get the desired coordinates for the phase space. We can formalize this idea as follows:
	
	\begin{definition}[Jacobi coordinates] \label{def:JacobiCoordinates} \quad \\
		Let $\C_1 = \C_{\wedge} \preccurlyeq \C_2 \preccurlyeq \ldots \preccurlyeq \C_n = \C_{\vee}$ be a strictly increasing sequence of cluster decompositions $C_i$ of $\bN$ with $\left|C_i\right| = n + 1 - i$ (this already follows from the strict increase of the ordering and the length of the sequence $n$). Then we know that for $i \in \{1, \ldots, n-1\}$ the set $C_{i} \setminus C_{i+1}$ consists of exactly two clusters, which we call $D_1^{i}, D_2^{i}$. Then \textbf{Jacobi coordinates} with respect to this sequence are any coordinates of the form
		\begin{align*}
			q_{D_1^{1}, D_2^{1}}^{I}, \ldots, q_{D_{1}^{n-1}, D_2^{n-1}}^{I}, q_{\bN}, p_{D_1^{1}, D_2^{1}}^{I}, \ldots, p_{D_{1}^{n-1}, D_2^{n-1}}^{I}, p_{\bN}.
		\end{align*}
	\end{definition}
	
	\begin{remark}
		\quad \\
		Note that we have a choice in the naming of the two clusters $D_1^{i}, D_2^{i}$ in the above definition. If we switch the naming for any index, this will give us a minus sign for both the positions and the momenta. This is why we do not speak about \enquote{the} Jacobi coordinates as we allow this freedom. \\
		Usually in applications, we will not give the sequence of cluster decompositions but rather just the set of Jacobi coordinates. \\
		Note that there are sometimes more than one sequence of cluster decompositions leading to the same Jacobi coordinates (with just a permutation of the coordinates as defined above). \hfill $\diamond$
	\end{remark}
	
	A good computational introduction to Jacobi coordinates where the previously described intuition of successively merging clusters is apparent is \cite{lim1991binary}, where the following theorems are proven:
	
	\begin{lemma}[Properties of Jacobi coordinates] \label{lem:PropJacobiCoordinates}\quad \\
		Let $q_{D_1^{1}, D_2^{1}}^{I}, \ldots, q_{D_{1}^{n-1}, D_2^{n-1}}^{I}, q_{\bN}, p_{D_1^{1}, D_2^{1}}^{I}, \ldots, p_{D_{1}^{n-1}, D_2^{n-1}}^{I}, p_{\bN}$ be Jacobi coordinates as in Definition \ref{def:JacobiCoordinates}.
		\begin{enumerate}[(i)]
			\item Then the Jacobi coordinates are symplectic coordinates, i.e. the linear map
			\begin{align*}
				&(p_1, \ldots, p_n, q_1, \ldots, q_n) \\ &\mapsto \left(p_{D_1^{1}, D_2^{1}}^{I}, \ldots, p_{D_{1}^{n-1}, D_2^{n-1}}^{I}, p_{\bN}, q_{D_1^{1}, D_2^{1}}^{I}, \ldots, q_{D_{1}^{n-1}, D_2^{n-1}}^{I}, q_{\bN}\right)
			\end{align*}
			is a symplectomorphism.
			\item Let $S \in \{L, J, K\}$ be a system quantity from Definition \ref{def:System_quantities}. Then $S$ can be expressed using the Jacobi coordinates as
			\begin{align*}
				S &= \sum_{i=1}^{n-1} S_{D_1^i, D_2^i}^{I} + S_{\bN}.
			\end{align*}
		\end{enumerate}
	\end{lemma}
	
	\begin{proof}
		\quad \\
		The first statement is just the theorem proved in \cite[p.157]{lim1991binary} whereas the second statement is expression (4.2) on page 163 from \cite{lim1991binary} up to some renaming of the quantities, some multiplicative factors and Remark~\ref{rem:ClusterCoordinates} that the unusual definition of the external cluster coordinates with respect to the total center of mass cancels out, since internal coordinates are used.
	\end{proof}

	\subsection{The final cluster decomposition and the asymptotic direction of non-collision singularities}
	
	This section is adapted from \cite[Section 2]{quaschner2025improbability}. Its purpose is to explain how one can decompose the set of particles $\bN$ into asymptotic clusters.
	\\
	One of the main questions when we are using some cluster coordinates is, how the clusters shall be defined. One idea for this is the so-called \textbf{Graf partition} in position space, from which a cluster decomposition can be deduced. We follow \cite[Chapter 12.6, p. 316 et seqq]{knauf2018mathematical} to recap these definitions and results that are only necessary for defining the final cluster decomposition in Definition~\ref{def:TempAndFinalClusterDecomposition} below:
	
	\begin{definition}[Graf partition] \label{def:GrafPartition} \quad \\
		For $\delta \in (0,1)$, let
		\begin{align*}
			J^{(\delta)}:& \; M \rightarrow [0, \infty), \quad J^{(\delta)}(q) := \max\left\{ J_{\C}^{E}(q) + \delta^{\left|\C \right|} \mid \C \in \mathcal{P}(\bN)\right\},
			\intertext{where}
			J_{\C}^{E}(q) &:= \sum_{C \in \C} \frac{m_C}{2} \norm{q_C}^2.
		\end{align*}
		This induces a measure-theoretic partition of the position space $M$ into the sets
		\begin{align*}
			\Xi_{\C}^{(\delta)} &:= \left\{ q \in M \mid J^{(\delta)}(q) = J_{\C}^{E}(q) + \delta^{\left|\C \right|} \right\}
		\end{align*}
		for arbitrary $\C \in \mathcal{P}(\bN)$. These sets $\left(\Xi_{\C}^{(\delta)}\right)_{\C \in \mathcal{P}(\bN)}$ are called the \textbf{Graf partition}. Alternatively, we can define for a position $q \in M$ any $\C \in \mathcal{P}(\bN)$ as a \textbf{Graf cluster decomposition}, if $\C$ satisfies $q \in \Xi_{\C}^{(\delta)}$.
	\end{definition}
	
	\begin{remark} \label{rem:PropertiesGrafPartition}
		\quad \\
		We gather some useful statements from \cite[Chapter 12.6]{knauf2018mathematical} on the Graf partition:
		\begin{enumerate}[(i)]
			\item The Graf partition is only a measure-theoretic partition of the phase space in the sense that $\Xi_{\C}^{(\delta)} \cap \Xi_{\mathcal{D}}^{(\delta)} \neq \emptyset$ is possible, but we always have
			\begin{align*}
				\lambda^{nd}\left(\Xi_{\C}^{(\delta)} \cap \Xi_{\mathcal{D}}^{(\delta)}\right) = 0 \quad \text{ for } \C \neq \mathcal{D}.
			\end{align*}
			\item If two particles are in the same cluster, their distance is bounded above by a constant that only depends on $\delta$ (and the masses).
			\item \label{bul:GrafPartitionMinimalDistance} If two particles are not in the same cluster, their distance is bounded below by a positive constant that only depends on $\delta$ (and the masses).
			\item \label{bul:DeltaGrafPartitionSufficientlySmall} For some sufficiently small $0<\delta_0< \frac{1}{n}$ and all $\delta \leq \delta_0$, if for some position $q$ there are two different Graf cluster decompositions $\C$ and $\mathcal{D}$, then these two are comparable in the partial order for clusters.
			\item For sufficiently small $\delta$ the distance between particles from the same cluster is always smaller than the distance between any particles from different clusters. \hfill $\diamond$
		\end{enumerate}
	\end{remark}
	
	With the Graf cluster decomposition we can define a temporal decomposition into clusters for solution curves of e.g. the Hamiltonian problem and determine the final splitting of the particles:
	
	\begin{definition}[Temporal and final Graf cluster decomposition] \label{def:TempAndFinalClusterDecomposition} \quad \\
		Let $T^{+} > 0$ be an arbitrary time and $q: [0, T^{+}) \rightarrow M$ be a continuous function.
		\begin{enumerate}[(i)]
			\item For arbitrary $\delta \in (0,1)$ we call a family $\left(\C_t^{(\delta)}\right)_{t \in [0, T^{+})}$ with $\C_t^{(\delta)} \in \mathcal{P}(\bN)$ a \textbf{temporal Graf cluster decomposition}, if
			\begin{align*}
				q(t) \in \Xi_{\C_t^{(\delta)}}^{(\delta)} \quad \forall t \in [0, T^{+}).
			\end{align*}
			\item For $\delta \in (0,\delta_0)$ as in Remark~\ref{rem:PropertiesGrafPartition}, \eqref{bul:DeltaGrafPartitionSufficientlySmall} we call a temporal Graf cluster decomposition $\left(\C_t^{(\delta)}\right)_{t \in [0, T^{+})}$  the \textbf{coarsest temporal Graf cluster decomposition}, if
			\begin{align*}
				\C_t^{(\delta)} &= \bigvee_{\substack{\C \in \mathcal{P}(\bN) \\ q(t) \in \Xi_{\C}^{(\delta)}}} \C.
			\end{align*}
			\item Let for $\delta \in (0,\delta_0)$ $\left(\C_t^{(\delta)}\right)_{t \in [0, T^{+})}$ be the coarsest temporal Graf cluster decomposition. We define the \textbf{final Graf cluster decomposition} $\C_{\text{fin}}$ as
			\begin{align*}
				\C_{\text{fin}} := \bigwedge_{\delta \in (0,\delta_0)} \bigwedge_{s \in (0, T^{+})} \bigvee_{t \in [s, T^{+})} \C_t^{(\delta)}.
			\end{align*}
		\end{enumerate}
	\end{definition}

	\begin{remark} \label{rem:AsymptoticGrafClusterDecomposition}
		\quad
		\begin{enumerate}[(a)]
			\item Note that there is not a unique temporal Graf cluster decomposition, as there is some freedom at times where two cluster decompositions attain the maximum in $J^{(\delta)}$. However, we can get a well-defined expression for small values $\delta \leq \delta_0$ as in the definition of the \textbf{coarsest} temporal Graf cluster decomposition.\footnote{Analogously one could have used the \textbf{finest} temporal cluster decomposition to dissolve this ambiguity.} This follows as for small $\delta$ due to the comparability of all possible cluster decompositions the largest and smallest element of this finite chain (as we only have finitely many clusters) are well defined and suitable cluster decompositions as well, i.e. we have indeed
			\begin{align*}
				q(t) \in \Xi_{\mathcal{D}}^{(\delta)} \quad \text{ for } \mathcal{D} := \bigvee_{\substack{\C \in \mathcal{P}(\bN) \\ q(t) \in \Xi_{\C}^{(\delta)}}} \C.
			\end{align*}
			\item The final Graf cluster decomposition provides information about which particles come arbitrarily close to each other at times close to the \enquote{escape} time $T^{+}$ (note that the definition works for arbitrary curves $q$ with values in $M$, but we used the notation $T^{+}$ to suggest the correspondence to the escape time in the setting of solution curves, which is our main application).
			\item Note that in the definition of $\C_{\text{fin}}$ the sets $\bigwedge_{s \in (0, T^{+})} \bigvee_{t \in [s, T^{+})} \C_t^{(\delta)}$ are indeed getting finer for smaller $\delta$-values, hence we could talk of the limit of the cluster decomposition for $\delta \downarrow 0$. 
			\item Since we have a finite lattice and a monotonically decreasing sequence \linebreak $\left(\bigwedge_{s \in (0, T^{+})} \bigvee_{t \in [s, T^{+})} \C_t^{(\delta)}\right)_{\delta \in (0, \delta_0)}$, there is some positive value $\delta_{\min} >0$ for which the minimum is attained, i.e. 
			\begin{align*}
				\C_{\text{fin}} = \bigwedge_{s \in (0, T^{+})} \bigvee_{t \in [s, T^{+})} \C_t^{(\delta_{\min})}.
			\end{align*}
			As the sequence $\left(\bigvee_{t \in [s, T^{+})} \C_t^{(\delta_{\min})}\right)_{s \in (0, T^{+})}$ is decreasing as well, by the same argument there is a time $s_1 < T^{+}$ with
			\begin{align*}
				\C_{\text{fin}} = \bigwedge_{s \in (0, T^{+})} \bigvee_{t \in [s, T^{+})} \C_t^{(\delta_{\min})} = \bigvee_{t \in [s_1, T^{+})} \C_t^{(\delta_{\min})}.
			\end{align*}
			So particles from different clusters of the final cluster decomposition will never be in the same temporal cluster decomposition after the time $s_1$ and hence (by the properties of the Graf partition, Remark~\ref{rem:PropertiesGrafPartition},\eqref{bul:GrafPartitionMinimalDistance}) will always have a positive distance (that depends on $\delta_{\min}$).\hfill $\diamond$
		\end{enumerate}
	\end{remark}
	
	Given the final cluster decomposition, we usually restrict our attention to the particles of a single final subsystem, i.e. a cluster from the final cluster decomposition. In this paper, we care about two different types of final clusters: Clusters suffering a total collision and clusters having a non-collision singularity (i.e. they become unbounded) and that consist of exactly four particles. For the latter case, we will use the following result that is proven in \cite[Section 4.4]{quaschner2025improbability}
	
	\begin{proposition}[Existence of an asymptotic direction] \label{prop:ExistenceAsymptoticDirection} \quad \\
		For a non-collision singularity of a subsystem $U$ with three clusters $\left(C_1^{(i)}, C_2^{(i)}, C_3^{(i)}\right)$ with infinitely many passages of the messenger $C_2^{(i)}$ from $C_1^{(i)}$ to $C_3^{(i)}$ or vice versa under some assumptions on the external forces the following limits exist and satisfy
		\begin{align*}
			\lim\limits_{i \rightarrow \infty} \hat{q}_{C_1^{(i)}}\left(x^{(i)}\right) &=- \lim\limits_{i \rightarrow \infty} \hat{q}_{C_3^{(i)}}\left(x^{(i)}\right).
		\end{align*}
		We call $\lim\limits_{i \rightarrow \infty} \hat{q}_{C_1^{(i)}}\left(x^{(i)}\right)$ the \textbf{asymptotic direction} of the non-collision orbit and denote it by $\operatorname{AD}_U$.
	\end{proposition}

	\subsection{The Poincaré Surface Method}
	\label{sec:PoincareSurfaceMethod} \quad
	
	The Poincaré surface method was developed in \cite{fleischer2018improbability} and generalized in \cite{fleischer2019improbability} in order to prove the improbability of some sets in phase space, e.g. the set of initial conditions leading to a collision singularity (compare \cite{fleischer2019improbabilityCollisions}), of initial conditions leading to non-collision singularities for $n=4$ bodies \cite{fleischer2018improbability} (with even more general assumptions than in this thesis) or of initial conditions leading to orbits where the asymptotic velocity does not exist \cite{knauf2018asymptotic}. \\
	This is our main tool in proving the improbability result of Theorem~\ref{thm:ImpResSubsystems}. In order to keep this paper self-contained we will recall some results following the presentation in \cite{fleischer2019improbability}.
	\\
	The general setting for the Poincaré surface method is a smooth manifold $P$ with a volume form $\Omega$ and a $C^{1}$ vector field $X$ for which the Lie derivative $\mathcal{L}_X \Omega = 0$, i.e. the flow generated by $X$ is volume preserving. This flow is denoted by $\Phi : D \rightarrow P$ and is defined on a maximal domain $D$ of the form
	\begin{align*}
		D := \left\{ (t,x) \in \R \times P \mid T^{-}(x) < t < T^{+}(x) \right\}
	\end{align*}
	where $T^{-}, T^{+} : P \rightarrow \R \cup \{\pm \infty\}$ are the escape times of the initial condition with $T^{+}>0, \; T^{-}<0$.
	
	\begin{definition}[Wandering points] \quad \\
		We call a point $x \in P$ \textbf{wandering} if and only if there exists a neighborhood $U_x$ of $x$ and some $t_{-} \in (0, T^{+}(x))$ with
		\begin{align*}
			U_x \cap \Phi(((t_{-}, T^{+}(x)) \times U_x) \cap D ) = \emptyset,
		\end{align*}
		i.e. after some time the neighborhood is not entered by the flow any more. The set of wandering points is denoted by $\operatorname{Wand}$.
	\end{definition}
	
	Wandering points are interesting for dynamical systems in general. The set of wandering points is in general not of $\Omega$-measure zero. But an important property is the following lemma \cite[Lemma 1.2]{fleischer2019improbability}:
	
	\begin{lemma}[Singularities and wandering points] \quad \\
		Let
		\begin{align*}
			\operatorname{Sing} := \{ x \in P \mid T^{+}(x)<\infty \}
		\end{align*}
		be the set of singular points. Then we have
		\begin{align*}
			\operatorname{Sing} \subseteq \operatorname{Wand}.
		\end{align*}
	\end{lemma}
	
	We now need a sequence of \textbf{Poincaré surfaces} $\iota_m: \overline{\mathcal{H}}_m \rightarrow P$ consisting of pairwise disjoint closed codimension-one $\partial$-submanifolds (which may or may not have a boundary, with $\mathcal{H}_m$ we will denote the manifold without its boundary) with two conditions:
	
	\begin{itemize}
		\item $X$ shall be \textbf{transversal} to $\mathcal{H}_m$. From this we obtain a $(\dim(P)-1)$-form  $i_X \Omega$ that is a volume form when restricted to  $\mathcal{H}_m$, more precisely we name this form \linebreak $\mathcal{V}_m := \iota_m^* (i_X \Omega)$.
		\item $\lim\limits_{m \rightarrow \infty} \int_{\mathcal{H}_m} \mathcal{V}_m = 0$, i.e. the volume of the surfaces converges to 0.
	\end{itemize}
	With this sequence of transversal submanifolds we can define another set of points:
	
	\begin{definition} [Transition points] \quad
		\begin{enumerate}[(i)]
			\item For an arbitrary initial condition $x \in P$ we define the \textbf{forward orbit} through $x$ as
			\begin{align*}
				O^{+}(x) &:= \{ \Phi(t,x) \mid 0 \leq t < T^{+}(x) \}.
			\end{align*}	
			\item We define the set of \textbf{transition points} that hit almost all of the Poin\-caré surfaces as
			\begin{align*}
				\operatorname{Trans}_{\Phi} := \left\{ x \in P \mid \exists m_0 \in \N, \; \forall m \geq m_0: O^{+}(x) \cap \overline{\mathcal{H}}_m \neq \emptyset \right\}.
			\end{align*}
		\end{enumerate}
	\end{definition}
	Note that, for a given family of Poincaré surfaces, the set of transition points need not have measure zero. But intersecting the set of transition points with the set of wandering points, we get a set of measure zero, as Theorem A in \cite{fleischer2019improbability} states:
	\begin{theorem*}[Improbability of the set of wandering transition points] \quad \\
		Under the above assumptions on the system and the Poincaré surfaces we have
		\begin{align*}
			\Omega(\operatorname{Trans}_{\Phi} \cap \operatorname{Wand}) = 0,
		\end{align*}
		where we identify a volume form with its induced measure. Therefore the union of all wandering orbits that consist of transition points to the Poincaré surface is a set of $\Omega$-measure zero.
	\end{theorem*}
	
	With this theorem we need to find a sequence of Poincaré surfaces for which our set of considered singularities is contained in the set of transition points. In order to prove the improbability of singularities in total, sometimes different Poincaré surfaces need to be defined. For example the improbability of collision singularities was proven with some sequence of Poincaré surfaces for arbitrary $n$ and proofs for the improbability of non-collision singularities could be using different Poincaré surfaces. This depends on the structure of the orbits corresponding to the different types of singularities.

	\subsection{Classical Results - the Theorems of Painlevé and von Zeipel}
	\label{sec:ClassicalResults} \quad
	
	The theorems of Painlevé and von Zeipel are classical results that qualitatively describe singularities in the $n$-body problem of celestial mechanics. Here we will use generalized versions that adapt the statement to subsystems of the final cluster decomposition. The proofs can be found in \cite[Sections 3.2,3.3]{duignan2026blowup}.
	
	\begin{theorem*}[Painlevé for subsystems] \quad \\
		Let $(q(t,x), p(t,x))_{t \in [0, T^{+}(x))}$ be a solution curve of the Hamiltonian system \eqref{eq:Hamiltonfunktion} with initial condition $x \in \operatorname{Sing}$, i.e. the escape time $T^{+}(x)$ is finite. Let $\C$ be the final cluster decomposition. Then we have for any $C \in \C$ with $\left|C\right|>1$:
		\begin{align*}
			\lim\limits_{t \uparrow T^{+}(x)} \min\left\{ \norm{q_i(t, x) - q_j(t, x)} \mid i, j \in C, i \neq j \right\} = 0.
		\end{align*}
	\end{theorem*}
	
	Painlevé's theorem provides some qualitative insight into the orbits of singularities, as the minimal distance has to converge to zero, i.e. there is at least one energy reservoir in the system. \\
	The second important qualitative result is that of von Zeipel which characterizes collision and non-collision singularities by their moment of inertia (compare Definition~\ref{def:System_quantities}). Additionally, we can adapt this statement for subsystems of the final cluster decomposition.
	
	\begin{theorem*}[von Zeipel] \quad \\
		Let $(q(t,x), p(t,x))_{t \in [0, T^{+}(x))}$ be a solution curve of the Hamiltonian system \eqref{eq:Hamiltonfunktion} with initial condition $x \in \operatorname{Sing}$. Let $\C$ be the final cluster decomposition. Then we have for any $C \in \C$ with $\left|C\right|>1$:
		\begin{align*}
			\lim\limits_{t \uparrow T^{+}(x)} J_C^{I}(q(t,x)) \in [0, \infty]
		\end{align*}
		exists and the limit is finite if and only if it is zero, in which case all particles of $C$ have a collision at some point in $\R^d$.
	\end{theorem*}
	
	Additionally, note that the center of mass of any cluster of the final cluster decomposition has a definite limit in $\R^d$. This is due to the fact that on the cluster momenta only bounded forces (from particles of the other clusters) can act and the total time is finite. Hence, subsystems having a total collision stay within a bounded set in $\R^d$, whereas the subsystems experiencing a non-collision singularities are unbounded (we will thus call them unbounded subsystems, as this is a shorter notation).
	\\
	From these classical results one can already get a qualitative course of non-collision singularities. It is not only necessary that some particles always have to be close together (compare Painlevé's theorem), but in order for the system to diverge, particles must come close to each other and then separate again. But close to the escape time, there has to be another interaction for these particles with different particles than before (a turnaround is only possible due to large forces, if there is only a small amount of time left).
	
	\section{Considerations of a colliding subsystem}
	\label{sec:CollidingSubsystems}
	
	The case of a cluster $B$ from the final cluster decomposition that has a bounded internal moment of inertia, i.e. $\lim\limits_{t \uparrow T^{+}} J_B^{I}(t) = 0$, was treated by Duignan and Quaschner \cite{duignan2026blowup}.
	For convenience we recall the important statements: The total change of the energy of the subsystem $h_B^{I}$ is bounded, hence also the total energy of the subsystem is bounded. From this one can derive that a perturbed collision orbit is similar to a collision orbit without perturbation. From this analysis, there was derived in \cite[Proposition 5.3]{duignan2026blowup} a set of finite volume in which the colliding subsystem will stay close to the escape time (depending on a bound $A$ for some system quantities): 
	\begin{proposition} \label{prop:SubsetCollisions}
		Let $A>0$ be given. Consider a system of $n$ particles $(q,p)$ evolving according to
		\begin{equation*}
			\dot{q}_j = \frac{p_j}{m_j},\qquad \dot{p}_j = -\frac{\partial V}{\partial q_j} + F_j(t),\qquad j = 1,\dots, n 
		\end{equation*}
		with attractive pair potentials and bounded external forces satisfying
		\[
		\sup_{t\in(0,T^+),\,j\in C} \norm{F_j(t)} 	\leq A .
		\]
		Then there is a set of finite volume $\mathcal{S}$ (depending on $A$), such that any collision solution satisfying
		\begin{equation}\begin{split}
				\lim_{t\uparrow T^+} J_C^I(t) = 0,\quad
				\limsup_{t\uparrow T^+}\norm{q_C(t)} \leq A,\quad
				\limsup_{t\uparrow T^+}\norm{p_C(t)} \leq A,\\
				\sup_{t\in(0,T^+)} |h_C^I(t)| \leq A,\quad 
				\limsup_{t\uparrow T^+} -V_C(Q(t))  \leq A
			\end{split}
		\end{equation}
		will stay in this set close to the escape time.
	\end{proposition}
	
	\section{Considerations of unbounded subsystems}
	\label{sec:UnboundedSubsystem}
	
	In this section, we consider an unbounded subsystem $U$ of the final cluster decomposition with exactly $\left|U\right| = 4$ particles. For notational simplicity, we will assume that $U=\{1,2,3,4\}$. The main analysis for such subsystems was carried out in \cite{quaschner2025improbability}. In order to apply the results from this paper, we have to check the following properties  \cite[Definition 3.1, 4.1]{quaschner2025improbability}
	
	\begin{definition}[Extended conditions for non-collision singularities of Xia-Type with four particles] \label{def:NC_SingXiaType} \quad \\
		Consider a system in which the equations of motion for positions $q$ and momenta $p$ for $n=4$ particles moving in $d$-dimensional space are given as
		\begin{align}
			\begin{split}
				\dot{q}_i(t) &= \frac{p_i(t)}{m_i}, \\
				\dot{p}_i(t) &= \sum_{\substack{j=1, \\ j \neq i}}^{n} \alpha Z_{i,j} \frac{q_i(t)- q_j(t)}{\norm{q_i(t)- q_j(t)}^{2+\alpha}} + F_i(t)
			\end{split}
			\label{eq:EquationsMotionsWithForces}
		\end{align}
		with interaction constants $Z_{i,j} = - m_i \cdot m_j$ 
		and \textbf{external forces}
		\begin{align}
			F_i(t) = \sum_{l=1}^{k} m_i W_l(q_i(t) - c_l(t)) \label{eq:ForcesAreCenterlike}
		\end{align}
		for some fixed $k \in \N$, continuous functions $c_l$ and functions
		\begin{align*}
			W_l \in C^{1}(\R^{d}\setminus\{0\}, \R^d)
		\end{align*}
		such that for any $\delta > 0$ the restrictions ${\left(W_l\right)}_{\mid \R^{d} \setminus B_{\delta}(0)}$ and the derivative ${\left(D W_l\right)}_{\mid \R^{d} \setminus B_{\delta}(0)}$ are bounded by some suitable constant $\operatorname{Const}_{W, \delta} \in (0, \infty)$ for any $l \in \{1, \ldots, k\}$. \\
		We call an initial condition $(q_0, p_0)$ a non-collision singularity of Xia-type, if there are two sequences of times $\left(x^{(i)}\right)_{i \in \N}$ and $\left(s^{(i)}\right)_{i \in \N}$ and a sequence of cluster decompositions $\C^{(i)} := \left\{ C_1^{(i)}, C_2^{(i)}, C_3^{(i)} \right\}$ of $\{1,2,3,4\}$ such that the following conditions are satisfied:
		\begin{enumerate}[(i)]
			\item \label{bul:PassageTimesMonotone}For all $i \in \N$ we have $x^{(i)} < s^{(i)} < x^{(i+1)} < T^{+}(q_0, p_0)$ and $\lim\limits_{i \rightarrow \infty} x^{(i)} = T^{+}(q_0, p_0)$. Here $T^{+}(q_0, p_0)<\infty$ is the escape time for these initial conditions.
			\item \label{bul:PassageTimesOddI} For odd $i$ we have
			\begin{align*}
				\norm{q_{C_1^{(i)}}\left(x^{(i)}\right) - q_{C_2^{(i)}}\left(x^{(i)}\right)} &= 1, \quad \norm{q_{C_3^{(i)}}\left(s^{(i)}\right) - q_{C_2^{(i)}}\left(s^{(i)}\right)} = 1 \\
				\inf_{t \in \left[x^{(i)}, s^{(i)}\right]}\norm{q_{C_1^{(i)}}\left(t\right) - q_{C_2^{(i)}}\left(t\right)} &\geq 1
			\end{align*}
			and $C_2^{(i)} \cup C_3^{(i)} = C_2^{(i+1)} \cup C_3^{(i+1)}$.
			\item \label{bul:PassageTimesEvenI}For even $i$ we have
			\begin{align*}
				\norm{q_{C_3^{(i)}}\left(x^{(i)}\right) - q_{C_2^{(i)}}\left(x^{(i)}\right)} &= 1, \quad \norm{q_{C_1^{(i)}}\left(s^{(i)}\right) - q_{C_2^{(i)}}\left(s^{(i)}\right)} = 1 \\
				\inf_{t \in \left[x^{(i)}, s^{(i)}\right]}\norm{q_{C_3^{(i)}}\left(t\right) - q_{C_2^{(i)}}\left(t\right)} &\geq 1
			\end{align*}
			and $C_2^{(i)} \cup C_1^{(i)} = C_2^{(i+1)} \cup C_1^{(i+1)}$.
			\item \label{bul:SmallDiameterClusters} We assume that the diameter of the nontrivial cluster is uniformly bounded by a small constant, i.e. there is some constant $\Con{sd} \ll 1$ such that
			\begin{align*}
				\sup_{i \in \N, t \in \left[x^{(i)}, s^{(i)}\right], j \in \{1,2,3\}, \left|C_j^{(i)}\right| = 2} \norm{q_{C_j^{(i)}}^{I}(t)} \leq \Con{sd}.
			\end{align*}
			\item \label{bul:FurtherForceConditions_GenAssumptionsPerturbedPassage} The following condition gives some bounds on the strength of the forces. We assume a minimal distance of at least $\delta>0$ for the centers of the disturbing forces to all particles, i.e.
			\begin{align*}
				\inf_{i \in \bN, l \in \{1, \ldots, k\}, t \in [0, T^{+}) } \norm{q_i(t) - c_l(t)} \geq \delta.
			\end{align*}
			By the equations of motion this leads to at most constant forces at all times up to the escape time.
			\item \label{bul:ConstDistDuringQuasicollision} In order to have a complete list of passages we demand that there is some constant such that for odd $i \in \N$
			\begin{align*}
				\sup_{t \in \left[s^{(i)}, x^{(i+1)}\right], j,k \in C_2^{(i)} \cup C_3^{(i)}} \norm{q_j(t) - q_k(t)} \leq \operatorname{Const},
				\intertext{whereas for even $i \in \N$:}
				\sup_{t \in \left[s^{(i)}, x^{(i+1)}\right], j,k \in C_2^{(i)} \cup C_1^{(i)}} \norm{q_j(t) - q_k(t)} \leq \operatorname{Const}.
			\end{align*}
			\item \label{bul:OuterClustersDoNotMeet}
			We assume that for all $i \in \N$ we have
			\begin{align*}
				\inf_{t \in \left[x^{(i)}, s^{(i)}\right], l \in C_1^{(i)}, k \in C_{3}^{(i)}} \norm{q_l(t) - q_k(t)} \geq 1.
			\end{align*}
			\item \label{bul:NC_Sing_Condition} We assume that
			\begin{align*}
				\limsup_{i \rightarrow \infty} \norm{q_{C_1^{(i)}}(x^{(i)})} = \infty.
			\end{align*}
			\item \label{bul:ConIntegratedForces} The forces integrated over a passage are always small, i.e. there is a constant $\Con{FI}$ such that for any $j \in \bN$ and $i \in \N$ we have
			\begin{align}
				\int_{x^{(i)}}^{x^{(i+1)}} \norm{F_j(t)} \diff t \leq  \frac{\operatorname{Const}_{FI}}{\max\left\{\norm{p_{C_1^{(i)}}\left(x^{(i)}\right)}, \norm{p_{C_3^{(i)}}\left(x^{(i)}\right)} \right\}}. \label{eq:ForceIntegratedSmall}
			\end{align}
			\item \label{bul:ConNegEnergyClusters} We assume that for all $i\in \N$ and $t \in \left[x^{(i)}, s^{(i)}\right]$ the internal energy of the non-trivial cluster is negative, i.e. for the cluster $D \in \left\{C_1^{(i)}, C_2^{(i)}, C_3^{(i)} \right\}$ with $\left|D\right|=2$ we have
			\begin{align*}
				H_D^{I}(t) < 0.
			\end{align*}
		\end{enumerate}
	\end{definition}
	
	We will briefly sketch how to interpret these conditions and why they are satisfied in our situation:
	\\
	The external forces $F_i$ are given by the influence of the particles in $\bN \setminus U$. They are of the desired form and by the final cluster decomposition, there is a minimal distance between particles from different clusters, which is exactly condition \eqref{bul:FurtherForceConditions_GenAssumptionsPerturbedPassage}. \\
	The conditions \eqref{bul:PassageTimesMonotone}, \eqref{bul:PassageTimesOddI}, \eqref{bul:PassageTimesEvenI}, \eqref{bul:ConstDistDuringQuasicollision} and \eqref{bul:OuterClustersDoNotMeet} are consistency conditions that are always satisfied for a non-collision singularity with four particles and are mainly used to fix the notation for the passages: The only mechanism for four bodies to become unbounded within finite time is a messenger cluster moving between two outer clusters (this was already known to Saari, the small forces do not perturb the proof). Then also \eqref{bul:NC_Sing_Condition} is just a weak formulation of von Zeipel's theorem (compare Section~\ref{sec:ClassicalResults}) for subsystems (we even get the statement $\lim_{i \rightarrow \infty} \norm{q_{C_1^{(i)}}(x^{(i)})} = \infty$, which could also be derived from all of these assumptions). Conditions \eqref{bul:SmallDiameterClusters} and \eqref{bul:ConNegEnergyClusters} are just a different formulation of Painlevé's theorem for subsystems from Section~\ref{sec:ClassicalResults}.
	\\
	Thus, it remains only to consider condition \eqref{bul:ConIntegratedForces} on a bound for the integrated forces. Note that this condition is in some sense stronger than just the boundedness of forces. We will use some heuristics based on the comparisons of different system quantities given in \cite[Definition 3.6, Corollary 3.7]{quaschner2025improbability} which do not need condition \eqref{bul:ConIntegratedForces} in the proofs:
	\begin{definition}[Comparison relations] \label{def:ComparisonRelations}\quad \\
		Let $a$ and $b$ be two time-dependent quantities. For any time interval $[t_1, t_2]$ we say that
		\begin{align*}
			a \lesssim b \text{ on } [t_1, t_2],
		\end{align*}
		if there is a constant $\Con{1}$, which only depends on the constants of the Hamilton function (masses, dimension etc.) such that
		\begin{align*}
			\forall t \in [t_1, t_2]: \quad a(t) \leq \Con{1} \cdot b(t).
		\end{align*}
		Usually, we do not stress the time interval $[t_1, t_2]$ explicitly. Additionally, we introduce the notation
		\begin{align*}
			a \sim b \; :\Leftrightarrow a \lesssim b \; \text{ and } b \lesssim a.
		\end{align*}
		We use both of these notations as well to compare two fixed values of time-dependent functions at different times.
	\end{definition}
	
	\begin{corollary}[Some comparisons] \label{cor:SomeComparisons} \quad \\
		Consider a non-collision singularity of Xia-type as in Definition~\ref{def:NC_SingXiaType}. For any odd $i>i_0$ we have the following comparisons:
		\begin{enumerate}[(i)]
			\item We can relate the time for a passage to the momenta and positions as
			\begin{align}
				\frac{\norm{q_{C_2^{(i)},C_3^{(i)}}^{I}\left(x^{(i)}\right)}}{s^{(i)} - x^{(i)}} \sim \norm{p_{C_2^{(i)},C_3^{(i)}}^{I}\left(x^{(i)}\right)}. \label{eq:VergleichOrtDurchZeitImpuls23}
			\end{align}
			\item We can compare the momenta in the time interval $\left[x^{(i)}, s^{(i)}\right]$:
			\begin{align}
				\begin{split}
					\norm{p_{C_2^{(i)},C_3^{(i)}}^{I}} &\lesssim \norm{p_{C_1^{(i)},C_2^{(i)}}^{I}} \sim \norm{p_{C_1^{(i)}}} \sim \norm{p_{C_2^{(i)}}}, \\ \norm{p_{C_3^{(i)}}} &\lesssim \norm{p_{C_2^{(i)}}}.
				\end{split}
				\label{eq:CompNormMomenta}
			\end{align}
			\item We have the following comparisons of different cluster distances in the time interval $\left[x^{(i)}, s^{(i)}\right]$:
			\begin{align}
				\norm{q_{C_1^{(i)}}} &\sim \norm{q_{C_3^{(i)}}} \sim \norm{q_{C_1^{(i)}}-q_{C_3^{(i)}}}. \label{eq:Vergleichq1q3q13I}
			\end{align}
		\end{enumerate}
		The constants for those comparisons can be chosen uniformly over all passages.
	\end{corollary}
	
	The duration of a passage is roughly $\operatorname{Const} \cdot \frac{\norm{q_{C_2^{(i)},C_3^{(i)}}^{I}\left(x^{(i)}\right)}}{\norm{p_{C_2^{(i)},C_3^{(i)}}^{I}\left(x^{(i)}\right)}}$, so if we integrate bounded forces over this time interval, we will have a too large contribution.  Note that \linebreak $\max\left\{\norm{p_{C_1^{(i)}}\left(x^{(i)}\right)}, \norm{p_{C_3^{(i)}}\left(x^{(i)}\right)} \right\}$ is approximately the scale with which the messenger cluster is moving which might be larger than $\norm{p_{C_2^{(i)},C_3^{(i)}}^{I}\left(x^{(i)}\right)}$. However, if we can prove that the forces are only of constant size within a small part of the passage and afterwards are decaying (e.g. as the distance between the particles grows), we can get better bounds for the integrated forces. 
	\\
	To prove this statement, we must consider the influence of another asymptotic cluster in Section~\ref{subsec:BoundsIntegratedForces}. There are two types of such clusters: Bounded clusters suffering a total collision and unbounded clusters also having a non-collision singularity. We will treat the case of a bounded cluster in Section~\ref{subsec:EstimatesForcesUnboundedAndBounded} and the case of two unbounded clusters using our assumption on the asymptotic direction in Section~\ref{sec:QuantitativeEstimatesDoubleConeArgument}. In Section~\ref{subsubsec:ConclusionsBoundExternalForces} we will recall some results from \cite{quaschner2025improbability} that are valid for these subsystems in some time interval of the passage. Finally, we will give some additional analysis on the motion in between the end of a passage with constant forces and up until the next passage starts, in Section~\ref{subsec:InvestigationMotionExtendedInterval} which is divided into suitable subsections considering different parts of the passage with different challenges. All of the above is based on some investigations in \cite[Chapters 4 and 5]{quaschner2023non}.
	
	\subsection{Bounds for the integrated forces}
	\label{subsec:BoundsIntegratedForces}
	
	We derive the bounds for the integrated forces and gather some important consequences of these bounds.
	
	\subsubsection{Estimates for the forces between an unbounded and a bounded subsystem}
	\label{subsec:EstimatesForcesUnboundedAndBounded}
	
	Consider an unbounded subsystem $U$ and another bounded subsystem $B$ of the final cluster decomposition. The following proposition establishes the force condition for the \enquote{external} forces of subsystem $B$ on particles of subsystem $U$.
	
	\begin{proposition}[Estimate for the integrated forces - Unbounded and bounded system]
		\label{prop:EstimatesIntegratedForcesUnboundedBounded} Consider an unbounded subsystem $U$ and a bounded subsystem $B$ from the final cluster decomposition of a singularity with $\left|U\right|= 4$. Then with the notation from Definition~\ref{def:NC_SingXiaType} we have for any passage with odd index $i$, sufficiently large $\norm{q_{C_1^{(i)}}\left(x^{(i)}\right)}$ and close to the escape time $T^{+}$ and arbitrary $k \in U$ and $j \in B$:
		\begin{align*}
			\int_{x^{(i)}}^{x^{(i+1)}} \norm{q_k(t)- q_j(t)}^{-1-\alpha} \diff t \leq  \frac{\operatorname{Const}_{FI}}{\max\left\{\norm{p_{C_1^{(i)}}\left(x^{(i)}\right)}, \norm{p_{C_3^{(i)}}\left(x^{(i)}\right)} \right\}}.
		\end{align*}
	\end{proposition}

	\begin{proof}\quad \\
		First of all, if we are sufficiently close to the escape time and $k \in C^{(i)}_l$, we can estimate $\norm{q_k(t) - q_j(t)} \geq \frac{1}{2} \norm{q_{C^{(i)}_l}(t) - q_B(t)}$, so it is enough to estimate the integral over the distance of the corresponding subcluster of $U$ containing $k$ and the center of mass of $B$. As $q_B$ stays within a bounded set, whereas $\norm{q_{C^{(i)}_1}}$ and $\norm{q_{C^{(i)}_3}}$ diverge, for sufficiently large $\norm{q_{C^{(i)}_1}(x^{(i)})}$ we have $\norm{q_{C^{(i)}_l}(t) - q_B(t)} \geq \operatorname{Const} \norm{q_{C^{(i)}_2}(t) - q_B(t)}$. So it is enough to estimate $\int_{x^{(i)}}^{s^{(i)}} \norm{q_{C^{(i)}_2}(t) - q_B(t)}^{-1-\alpha} \diff t$.  For the interval $\left[s^{(i)}, x^{(i+1)}\right]$ one can use a similar argument considering $\norm{q_{C^{(i)}_1}(t) - q_B(t)}^{-1-\alpha}$ and that these clusters separate with a velocity roughly $\norm{p_{C_1^{(i)}}\left(x^{(i)}\right)}$. \\
		For investigating $\int_{x^{(i)}}^{s^{(i)}} \norm{q_{C^{(i)}_2}(t) - q_B(t)}^{-1-\alpha} \diff t$, note that $f: t \mapsto \norm{q_{C^{(i)}_2}(t) - q_B(t)}^2$ is a convex function on $\left[x^{(i)}, s^{(i)}\right]$, as the second derivative is (up to small error terms and maybe some mass factors) given by $K_{C_2^{(i)}, B}^{I}$. As $p_B$ is roughly constant, this is approximately (up to mass factors) $\norm{p_{C_2^{(i)}}\left(x^{(i)}\right)}^2$, hence large. Thus there is a unique time $w_{\min}$ where $f$ is minimal on the interval $\left[x^{(i)}, s^{(i)}\right]$. Also, $f(v_{\min})$ is bounded by a constant depending on the minimal distance of the particles from different subsystems. From this time on we can use a propagation estimate, which yields the claim, e.g.
		\begin{align*}
			&\int_{w_{\min}}^{s^{(i)}} \norm{q_{C^{(i)}_2}(t) - q_B(t)}^{-1-\alpha} \diff t \leq \\&\leq \operatorname{Const} \int_{w_{\min}}^{s^{(i)}}\left(\norm{q_{C^{(i)}_2}(w_{\min}) - q_B(w_{\min})}^2 + \operatorname{Const} \norm{p_{C_2^{(i)}}\left(x^{(i)}\right)}^2 (t-w_{\min})^2 \right)^{-\frac{1+\alpha}{2}} \diff t \\
			& \leq \frac{\operatorname{Const}}{\norm{q_{C^{(i)}_2}(w_{\min}) - q_B(w_{\min})} \cdot \norm{p_{C_2^{(i)}}\left(x^{(i)}\right)}}.
		\end{align*}
		All other integrals can be estimated in an analogous way and using the comparison relations, we obtain the statement.
	\end{proof}

	\subsubsection{Estimates for the forces between two unbounded subsystems}
	\label{sec:QuantitativeEstimatesDoubleConeArgument}
	
	In the preceding subsection we demonstrated, how the integrated forces between an unbounded and a bounded cluster can be estimated. The main point was that there is only a small \enquote{interaction zone} where there are forces of constant size, because particles from the bounded cluster are always close to their limiting position, and that the particles separated with a large velocity. Both of these assumptions need not be the case if we consider two unbounded clusters: The interaction zone could be arbitrary large and the relative velocity of particles from different subsystems could be smaller than some quantities relevant for the orbit of one of the subsystems\footnote{To be more precise: Note that for the bounded subsystem, we could ignore the velocity of the center of mass of this subsystem, as it was only of constant size. In this case, if we consider the relative velocity between two centers of mass, there might be some cancellation, as both quantities could be comparable.}. However, for this to happen both subsystems would need to have the same asymptotic direction. Under our assumption of different asymptotic directions, we will be able to exclude this behavior the same way as in \cite[Chapter 5]{quaschner2023non}. For simplicity, we will use the following naming convention for the two unbounded systems of four particles each. The first subsystem shall be $\{1,2,3,4\}$, the second subsystem is $\{5,6,7,8\}$. It will be enough to give an argument on the size of the forces from the second subsystem acting on the first subsystem. Our general assumption is, that for a fixed small value $\beta > 0$, for the asymptotic direction of Proposition~\ref{prop:ExistenceAsymptoticDirection} the following holds:
	\begin{align*}
		\norm{\operatorname{AD}_{\{1,2,3,4\}}- \operatorname{AD}_{\{5,6,7,8\}}} &\geq \beta, \quad & \norm{\operatorname{AD}_{\{1,2,3,4\}} + \operatorname{AD}_{\{5,6,7,8\}}} &\geq \beta.
	\end{align*}
	
	The asymptotic directions do not completely define a region where the particles of one subsystem need to be close to the escape time. First of all, we have to consider directions within a neighborhood of the asymptotic directions too, as the latter are only the limit. The second problem is that the asymptotic directions are defined with respect to the common center of mass of the subsystem. This, however, stays within a ball of constant size, depending on some bound $\mu_q$ for each center of mass. In order to find a first approximating set in which one subsystem will be (in the position space), we have to take these two effects into account. We arrive at some sort of double cone with thickened top (see as well Figure~\ref{fig:SketchDoubleCone}):
	
	\begin{figure}[htbp]
		\includegraphics[width=\textwidth]{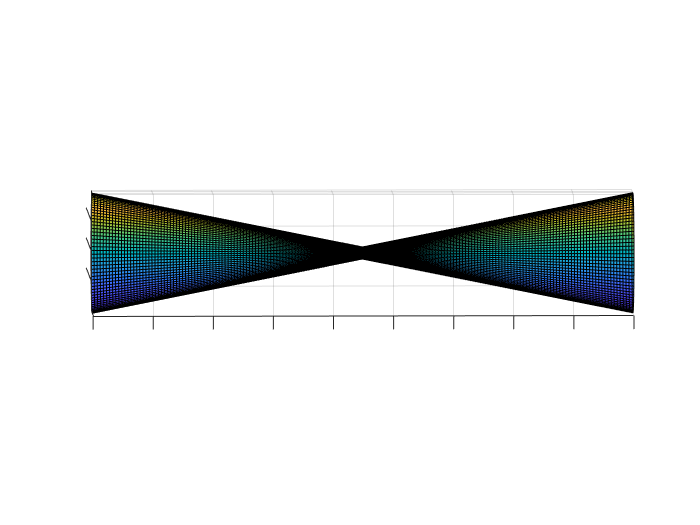}
		\caption[Sketch of a double cone]{A sketch of a double cone as in Definition~\ref{def:DoubleCones}. One can see that the tip of the cone is thickened, which corresponds to the parameter $\mu$. The reference direction $\theta$ is here the vertical direction. The parameter $\phi$ corresponds to the angle of the cone.}
		\label{fig:SketchDoubleCone}
	\end{figure}

	\begin{definition}[Double cones] \label{def:DoubleCones}\quad \\
		Assume that $\mu>0, \phi>0$ are some fixed positive numbers and $\theta \in \mathcal{S}^{d-1}$ is a fixed direction. The \textbf{$(\mu,\phi)$-double cone with direction $ \theta$} is the following subset of $\R^d$:
		\begin{align*}
			\operatorname{DC}_{\mu, \phi}(\theta)&:= \overline{B_\mu^{d}}\left(\left\{ q \in \R^d\setminus\{0\} \mid \exists \sigma \in \{-1,1\}: \norm{\theta - \sigma\hat{q}} \leq \phi \right\} \cup \{0\} \right).
		\end{align*}
		Here, for a subset $A\subseteq \R^{d}$ we define $\overline{B_{\mu}^{d}}(A)$ to be the closed $\mu$-neighborhood of the set $A$, more formally:
		\begin{align*}
			\overline{B_{\mu}^{d}}(A) = \{q \in \R^{d} \mid \operatorname{dist}(q, A) \leq \mu \}.
		\end{align*}
	\end{definition}

	\begin{remark} \label{rem:ExplanationDoubleCone}\quad 
		\begin{enumerate}[(i)]
			\item Note that the $\operatorname{DC}_{0, \phi}(\theta)$-double cone is indeed a geometric double cone. The parameter $\mu$ allows some thickening of this cone in the above defined sense.
			\item It is easy to see that the following is an alternative equivalent definition of the double cone:
			\begin{align*}
				\operatorname{DC}_{\mu, \phi}(\theta)=\bigg\{ q \in \R^d \mid &\exists \tilde{q} \in \overline{B_{\mu}^{d}}(0)\setminus\{q\}, \sigma \in \{-1,1\}: \\ & \norm{\theta - \sigma\frac{q-\tilde{q}}{\norm{q-\tilde{q}}}} \leq \phi \bigg\} \cup \left\{0\right\}.
			\end{align*}
			Here $\overline{B_{\mu}^{d}}(0)$ denotes the closed ball around $0\in \R^{d}$ with radius $\mu$. In this formulation we can reinterpret the parameter $\mu$ for our subsystem: Usually we consider the coordinates of a subsystem relative to the center of mass of this subsystem. Most of the time we cannot say, where the center of mass is, but we know it is within a ball of radius $\mu_q$ (where $\mu_q$ is a general bound for every center of mass of a subsystem). Thus if we find some suitable angle bounds for the relative coordinates, we can say that the absolute coordinates are within some double cone. This is one reason, why the double cones are important for the following considerations.
			\item We allow the sign $\sigma$ in both forms of the definition as we want to go in both directions of $\theta$, i.e. the whole one-dimensional subspace spanned by it. In some sense it would thus be more natural to take $\theta$ from the projective space and to replace the euclidean norm of the directions by the natural distance in this space. However, we do not believe that this is a simpler way of representing the results, but therefore we will often need to make assumptions on $\norm{\theta_1 \pm \theta_2}$ as in the lemmas below.
			\hfill $\diamond$
		\end{enumerate}
	\end{remark}
	
	We formulate some simple properties of double cones that will become useful later:
	
	\begin{lemma}[Nesting of double cones] \label{lem:NestingDoubleCones} \quad \\
		Let $\theta_1, \theta_2 \in \Sd$, $\phi, \psi >0$ be arbitrary with $\norm{\theta_1- \theta_2} \leq \psi$ or $\norm{\theta_1 + \theta_2} \leq \psi$. Then we have
		\begin{align*}
			\operatorname{DC}_{\mu, \phi}(\theta_2) \subseteq \operatorname{DC}_{\mu, \phi+\psi}(\theta_1).
		\end{align*}
	\end{lemma}
	
	\begin{proof}
		\quad \\
		Let $a \in \operatorname{DC}_{\mu, \phi}(\theta_2)$ be arbitrary. By the alternative definition there is a $q_a \in \overline{B_{\mu}^{d}(0)}$ and a $\sigma_a \in \{-1,1\}$ with
		\begin{align*}
			\norm{\theta_2 - \sigma_a\frac{a-q_a}{\norm{a-q_a}}} \leq \phi.
		\end{align*}
		But using the triangle inequality we immediately get
		\begin{align*}
			\norm{\theta_1 - \sigma_a\frac{a-q_a}{\norm{a-q_a}}} \leq
			\norm{\theta_2 - \sigma_a\frac{a-q_a}{\norm{a-q_a}}} + \norm{\theta_1- \theta_2},
		\end{align*}
		or 
		\begin{align*}
			\norm{\theta_1 + \sigma_a\frac{a-q_a}{\norm{a-q_a}}} \leq
			\norm{\theta_2 - \sigma_a\frac{a-q_a}{\norm{a-q_a}}} + \norm{\theta_1 + \theta_2}.
		\end{align*}
		One of these is smaller than $\phi+\psi$ by the assumption, hence we obtain $a \in \operatorname{DC}_{\mu, \phi+\psi}(\theta_1)$ and thus the claim.
	\end{proof}

	\begin{lemma}[Minimal distance between positions from different cones] \label{lem:MinDistDifferentCones} \quad \\
		Let $\theta_1, \theta_2 \in \Sd$, $\phi>0$ and $\mu>0$ be arbitrary with $\norm{\theta_1- \theta_2} \geq \phi$ and $\norm{\theta_1 + \theta_2} \geq \phi$. Then there are constants $\Con{\mu, \phi} > 0$ and $R_0(\mu, \phi)>0$ such that for any $a \in \operatorname{DC}_{\mu, \frac{\phi}{4}}(\theta_1)$ and $b \in \operatorname{DC}_{\mu, \frac{\phi}{4}}(\theta_2)$ with $\norm{a} \geq R_{0}(\mu, \phi)$ we have:
		\begin{align*}
			\norm{a-b} \geq \Con{\mu, \phi} \norm{a}.
		\end{align*}
	\end{lemma}
	
	\begin{proof}
		\quad \\
		Let $a,b$ be as in the statement. First of all we can assume that $\Con{\mu, \phi}\leq \frac{1}{2}$, but we will fix its value and the value of $R_{0}(\mu, \phi)$ within the proof. \\
		If we have $\norm{b} \leq \frac{R_{0}(\mu, \phi)}{2}$, we get
		\begin{align*}
			\norm{a-b} \geq \norm{a} - \norm{b} \geq \Con{\mu, \phi} \norm{a}.
		\end{align*}
		So we can assume that $\norm{b} \geq \frac{R_{0}(\mu, \phi)}{2}$ is large as well. In this case we start with a comparison of the directions $\hat{a}$ and $\hat{b}$.
		\begin{figure}[htbp]
			\includegraphics[width=\textwidth]{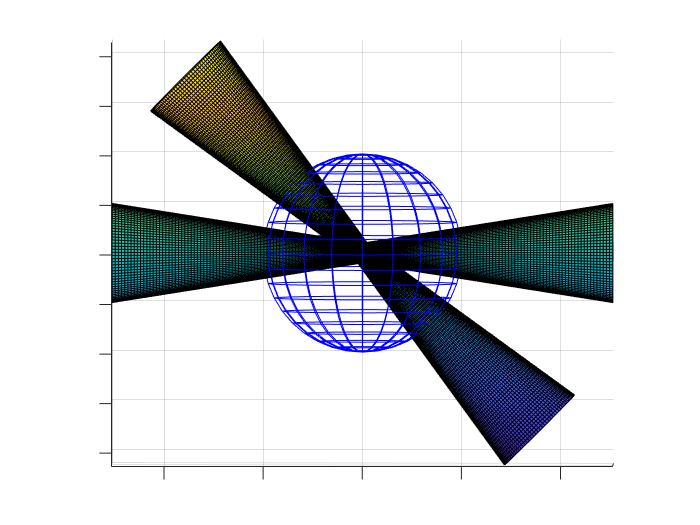}
			\caption[Angle between two double cones]{A graphical depiction of two double cones for the proof of Lemma~\ref{lem:MinDistDifferentCones}. The blue ball represents the ball of radius $R_0$ outside of which we have a minimal angle to points from the second double cone with sufficiently large norm as well.}
			\label{fig:TwoDoubleConesAndBall}
		\end{figure}
		Consider Figure~\ref{fig:TwoDoubleConesAndBall} for a graphical depiction of the idea. For that we fix $q_a, q_b \in \overline{B_{\mu}^{d}(0)}$ and $\sigma_a, \sigma_b \in \{-1,1\}$ such that
		\begin{align*}
			\norm{\theta_1 - \sigma_a\frac{a-q_a}{\norm{a-q_a}}} &\leq \frac{\phi}{4}, \quad &
			\norm{\theta_2 - \sigma_b\frac{b-q_b}{\norm{b-q_b}}} \leq \frac{\phi}{4}.
		\end{align*}
		Then we can estimate:
		\begin{align*}
			\norm{\hat{a} \pm \hat{b}} &\geq -\norm{\hat{a}-\frac{a-q_a}{\norm{a-q_a}}} - \norm{\frac{a-q_a}{\norm{a-q_a}} - \sigma_a \theta_1} + \norm{\sigma_a\theta_1 \pm \sigma_b \theta_2} \\ & \quad - \norm{\sigma_b \theta_2 - \frac{b-q_b}{\norm{b-q_b}}} - \norm{\frac{b-q_b}{\norm{b-q_b}} - \hat{b}} \geq \\
			& \geq \frac{\phi}{2} - \norm{\hat{a}-\frac{a-q_a}{\norm{a-q_a}}} - \norm{\frac{b-q_b}{\norm{b-q_b}} - \hat{b}}.
		\end{align*}
		Here we used our assumption on $\norm{\theta_1 \pm \theta_2}$ and the definitions of $q_a, q_b$. The last step consists now in considering
		\begin{align*}
			\norm{\hat{a}-\frac{a-q_a}{\norm{a-q_a}}} &\leq \frac{2 \norm{q_a}}{\norm{a-q_a}} \leq \frac{2 \mu}{R_{0}(\mu, \phi) - \mu}.
		\end{align*}
		A similar estimate holds for $\norm{\frac{b-q_b}{\norm{b-q_b}} - \hat{b}}$, as we assumed that  $\norm{b} \geq \frac{R_{0}(\mu, \phi)}{2}$. By choosing $R_{0}(\mu, \phi)$ sufficiently large, we get that
		\begin{align*}
			\norm{\hat{a} \pm \hat{b}} &\geq \frac{\phi}{4},
		\end{align*}
		i.e. there is a fixed $\phi$-dependent minimal angle between $a$ and $b$. But now we can use the law of sines to relate the distance $\norm{a-b}$ to $\norm{a}$, and thus get the statement by choosing $\Con{\mu, \phi}$ appropriately.
	\end{proof}
	
	For a subsystem $C \in \{\{1,2,3,4\}, \{5,6,7,8\} \}$ the interesting double cone will be of the form  $\operatorname{DC}_{\mu, \phi}(\operatorname{AD}_{C})$ for some suitable constants $\mu$ and $\phi$. We will show now that these cones with an arbitrary $\phi>0$ and suitable $\mu$ cannot be left as long as we are close to the escape time: 
	
	\begin{lemma}[Subsystems move in double cones] \label{lem:subsysDoubleCones} \quad \\
		Consider an unbounded subsystem $C$ of the final cluster decomposition. \\
		For arbitrary $\phi>0$, after some $t_0$ sufficiently close to the escape time and for some mass-dependent constant $\Con{\text{Cone}}$ we have for all $i \in C$ and $t \in (t_0, T^{+})$:
		\begin{align*}
			q_i(t) \in \operatorname{DC}_{\mu_q + \Con{\text{Cone}}, \phi}(\operatorname{AD}_{C}).
		\end{align*}
	\end{lemma}
	
	\begin{proof}
		\quad \\
		Let us consider $C= \{1,2,3,4\}$ and an arbitrary $\phi$. We consider passages with odd index sufficiently close to $T^{+}$ such that at the initial time $x^{(k)}$ of the $k$th passage we always have
		\begin{align*}
			\norm{\hat{q}_{C_1^{(k)}}\left(x^{(k)}\right) - \operatorname{AD}_{\{1,2,3,4\}}} \leq \frac{\phi}{4}.
		\end{align*}
		We will now consider two passages from time $x^{(k)}$ up to time $x^{(k+2)}$ and investigate, how far the particles can deviate from the line through $\hat{q}_{C_1^{(k)}}\left(x^{(k)}\right)$. \\
		First of all note that at time $x^{(k)}$ we have
		\begin{align*}
			\norm{\left(p_{C_1^{(k)}}\left(x^{(k)}\right)\right)^{\perp q_{C_1^{(k)}}\left(x^{(k)}\right)} } = \frac{\norm{L_{C_1^{(k)}}\left(x^{(k)}\right)}}{\norm{q_{C_1^{(k)}}\left(x^{(k)}\right)}} \leq \operatorname{Const}
		\end{align*}
		using \cite[Lemma 4.20, p. 58]{quaschner2025improbability}. Hence, by integrating the equations of motion, we have at any time $t \in \left[x^{(k)}, s^{(k+1)}\right]$:
		\begin{align*}
			\norm{\left(q_{C_1^{(k)}}\left(t\right)\right)^{\perp q_{C_1^{(k)}}\left(x^{(k)}\right)} } \leq \operatorname{Const} \cdot\left(s^{(k+1)} - x^{(k)}\right). 
		\end{align*}
		This is small and by a center of mass argument it suffices now to consider one other cluster. For the first passage up to time $s^{(k)}$ we take $C_2^{(k)}$. We know
		\begin{align*}
			\norm{\left(q_{C_2^{(k)}}\left(x^{(k)}\right)\right)^{\perp q_{C_1^{(k)}}\left(x^{(k)}\right)} } &\leq 1, \\
			\norm{\left(q_{C_2^{(k)}}\left(s^{(k)}\right)\right)^{\perp q_{C_1^{(k)}}\left(x^{(k)}\right)} } &\leq 1 + \norm{\left(q_{C_2^{(k)} \cup C_3^{(k)}}\left(s^{(k)}\right)\right)^{\perp q_{C_1^{(k)}}\left(x^{(k)}\right)} } = \\ & = 1 + \operatorname{Const} \norm{\left(q_{C_1^{(k)}}\left(s^{(k)}\right)\right)^{\perp q_{C_1^{(k)}}\left(x^{(k)}\right)} } \leq \\ & \leq 1 +o(1).
		\end{align*}
		But as $\left(p_{C_2^{(k)}}\left(x^{(k)}\right)\right)^{\perp q_{C_1^{(k)}}\left(x^{(k)}\right)}$ only changes little during this time interval due to the at most constant forces, we must have for all $t \in \left[x^{(k)}, s^{(k)}\right]$
		\begin{align*}
			\norm{\left(q_{C_2^{(k)}}\left(t\right)\right)^{\perp q_{C_1^{(k)}}\left(x^{(k)}\right)} } &\leq \operatorname{Const}.
		\end{align*}
		For times $t \in \left[s^{(k)}, x^{(k+1)}\right]$ the claim is clear due to the center of mass frame, as all particles from $C_2^{(k)} \cup C_3^{(k)}$ are at most a constant apart from the common center of mass. \\
		For the interval $\left[x^{(k+1)}, s^{(k+1)}\right]$ essentially the same argument works as for the interval $\left[x^{(k)}, s^{(k)}\right]$, but considering the possibly new cluster $C_2^{(k+1)}$ and the according expression $\left(q_{C_2^{(k+1)}}\right)^{\perp q_{C_1^{(k)}}\left(x^{(k)}\right)}$. \\
		Finally, during the interval $\left[s^{(k+1)}, x^{(k+2)}\right]$ we can again use the small distance from the common center of mass. \\
		Hence, we have established that for all $t \in \left[x^{(k)}, x^{(k+2)}\right]$ and $i \in \{1,2,3,4\}$ we have:
		\begin{align*}
			\operatorname{dist}\left(q_{\{i\}}(t), \operatorname{span}\left\{q_{C_1^{(k)}}\left(x^{(k)}\right)\right\}\right) \leq \Con{\text{Cone}}
		\end{align*}
		for some fixed mass-dependent constant $\Con{\text{Cone}}$ that could be calculated from the above estimates. But from this the existence of some $a_i(t) \in \R^{d}$ with $\norm{a_i(t)} \leq \mu_q + \Con{\text{Cone}}$ and
		\begin{align*}
			q_i(t) - a_i(t) \in \operatorname{span}\left\{q_{C_1^{(1)}}\left(x^{(1)}\right)\right\}
		\end{align*}
		follows. This, however, implies that
		\begin{align*}
			q_i(t) \in \operatorname{DC}_{\mu_q + \Con{\text{Cone}}, \frac{\phi}{4}}(\hat{q}_{C_1^{(1)}}\left(x^{(1)}\right)) \subseteq \operatorname{DC}_{\mu_q + \Con{\text{Cone}}, \phi}(\operatorname{AD}_{C}).
		\end{align*}
		Note that we would not even need $\frac{\phi}{4}$ for the first statement, we could even choose 0. The inclusion of the cones follows from Lemma~\ref{lem:NestingDoubleCones}, where even the cone $\operatorname{DC}_{\mu_q + \Con{\text{Cone}}, \frac{\phi}{4}}(\operatorname{AD}_{C})$ would be enough. As we considered an arbitrary double passage as long as it was close enough to the escape time, the claim follows.
	\end{proof}
	
	With these preparations we are able to prove that the integrated forces between the subsystems are sufficiently small during the passage of one subsystem:
	
	\begin{proposition}[Bound on the integrated forces between the subsystems during one passage] \label{prop:BoundIntegratedForcesOnePassage} \quad \\
		Consider two unbounded subsystems $U_1, U_2 \in \C_{\text{fin}}$ with four particles. Then there is a constant $\Con{F,\text{Cones}}$ depending on all constants of the Hamilton function and the fixed quantities describing the shape of the orbit, especially $\mu_q$ for the center of masses, $\delta$ for the minimal distance between particles from different subsystems, and $\beta$ for the angle between the asymptotic directions, such that for any passage $k$ of the subsystem $U_1$ sufficiently close to the escape time and any $i \in U_1$, $j \in U_2$ we have:
		\begin{align*}
			\int_{x^{(k)}}^{x^{(k+1)}} \norm{q_i(t)-q_j(t)}^{-(1+\alpha)} \diff t \leq \frac{\Con{F,\text{Cones}}}{\max\left\{ \norm{p_{C_1^{(k)}} \left(x^{(k)}\right)}, \norm{p_{C_3^{(k)}} \left(x^{(k)}\right)} \right\}}.
		\end{align*}
		Hence, the integrated forces between the systems are bounded in a suitable way.
	\end{proposition}
	
	\begin{proof}
		\quad \\
		Using Lemma~\ref{lem:subsysDoubleCones} we can assume that in all considered passages both systems always stay within $\operatorname{DC}_{\mu_q + \Con{\text{Cone}}, \frac{\beta}{4}}(\operatorname{AD}_{U_i})$, $i=1,2$. By the assumption on the asymptotic directions we know that
		\begin{align*}
			\norm{\operatorname{AD}_{U_1}- \operatorname{AD}_{U_2}} &\geq \beta, \quad & \norm{\operatorname{AD}_{U_1} + \operatorname{AD}_{U_2}} &\geq \beta.
		\end{align*}
		So Lemma~\ref{lem:MinDistDifferentCones} is applicable and gives us some constants
		\begin{align*}
			R_0 := R_0(\mu_q + \Con{\text{Cone}}, \beta), \quad \text{and } 
			\Con{1} := \Con{\mu_q + \Con{\text{Cone}}, \beta}
		\end{align*}
		such that we can bound from below the distance between points from the different cones as long as one of them is large. Now we will prove the statement for the first subsystem $U_1=\{1,2,3,4\}$. \\
		First of all we should note that it is enough to consider particles from $C_2^{(k)}$. This is due to the fact that by the comparison relations from Corollary~\ref{cor:SomeComparisons} there are mass-dependent constants such that
		\begin{align*}
			\norm{q_{C_1^{(k)}}} \geq \operatorname{Const} \norm{q_{C_2}^{(k)}}, \quad  \norm{q_{C_3^{(k)}}} \geq \operatorname{Const} \norm{q_{C_2}^{(k)}}.
		\end{align*}
		With this estimate the analogous reasoning as below can be applied to particles from $C_1^{(k)}$ or $C_3^{(k)}$. \\
		Next we split the passage from $x^{(k)}$ to $s^{(k)}$ further. Denote with $u_{\text{in}}$ and $u_{\text{out}}$ the first resp. last time with
		\begin{align*}
			\norm{q_{C_2^{(k)}}\left(u_{\text{in}}\right)} = 2\left(R_0 + \mu_q\right), \quad \norm{q_{C_2^{(1k)}}\left(u_{\text{out}}\right)} = 2\left(R_0 + \mu_q\right),
		\end{align*} 
		if $\norm{q_{C_2^{(k)}}} \leq 2\left(R_0 + \mu_q\right)$ holds at some time in the interval. 
		\begin{figure}[htbp]
			\includegraphics[width=\textwidth]{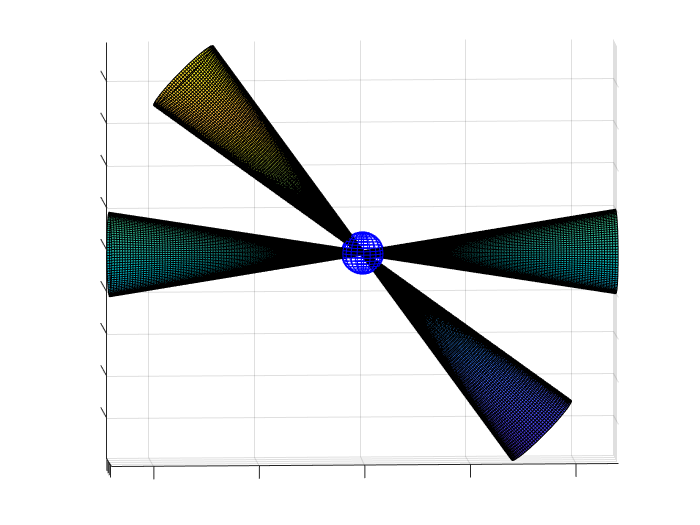}
			\caption[Interaction zone for the two subsystems]{Two double cones for the respective subsystems and the blue ball depicting the (small) interaction zone. $u_{\text{in}}$ and $u_{\text{out}}$ from Proposition~\ref{prop:BoundIntegratedForcesOnePassage} are the times when $q_{C_2^{(k)}}$ is entering this zone. Inside the zone we have constant force bounds, whereas outside of it we can use propagation estimates.}
			\label{fig:TwoDoubleConesSmallInteractionZone}
		\end{figure}
		For a schematic illustration, we refer to Figure~\ref{fig:TwoDoubleConesSmallInteractionZone}.\\
		For the uniqueness of these times we can argue as follows: First of all close to the escape time we know that $\norm{q_{C_2^{(k)}}\left(x^{(k)}\right)} > 2\left(R_0 + \mu_q\right)$ has to hold. Then one can easily see that
		\begin{align*}
			\ddot{J}_{C_2^{(k)}} \geq \operatorname{Const} \norm{p_{C_2}^{(k)}\left(x^{(k)}\right)}^2
		\end{align*}
		in this interval, hence the function is convex. Finally we can assume that $\norm{q_{C_2^{(k)}}} \leq 2\left(R_0 + \mu_q\right)$ has to hold at some time, otherwise the argument is even simpler and we would only need to consider the time of the global minimum of $\norm{q_{C_2^{(k)}}}$ for the following analysis, consider the case of a bounded subsystem in Proposition~\ref{prop:EstimatesIntegratedForcesUnboundedBounded}. \\
		The idea is as follows: Within the time interval $\left[u_{\text{in}}, u_{\text{out}}\right]$ we can only use the constant distance bound between particles in $C_2^{(k)}$ and particles from the other subsystem. Hence, we have to ensure that this time interval is short. Outside of this interval, we can relate the distance to $\operatorname{Const} \norm{q_{C_2^{(k)}}}$ using Lemma~\ref{lem:subsysDoubleCones} and then we can estimate the integral using a propagation estimate.\\
		First we consider the interval $\left[u_{\text{in}}, u_{\text{out}}\right]$. Using a simple propagation estimate (from the time of a global minimum of $J_{C_2^{(k)}}$ during this passage in both time directions), we get that
		\begin{align*}
			u_{\text{out}} -   u_{\text{in}} \leq \frac{\operatorname{Const} \cdot 4\left(R_0 + \mu_q\right)}{\norm{p_{C_2^{(k)}}\left(x^{(k)}\right)}}.
		\end{align*}
		Hence, we get for any $i \in C_2^{(k)}$ and $j \in U_2$:
		\begin{align*}
			\int_{u_{\text{in}}}^{u_{\text{out}}} \norm{q_i(t)-q_j(t)}^{-(1+\alpha)} \diff t \leq \frac{\operatorname{Const} \cdot 4\left(R_0 + \mu_q\right) \delta^{-1-\alpha}}{\norm{p_{C_2}\left(x^{(1)}\right)}}.
		\end{align*}
		Next we consider the time interval $\left[x^{(1)}, u_{\text{in}}\right]$, the interval $\left[u_{\text{out}}, s^{(1)}\right]$ can be treated analogously. First of all we notice that for these times we have
		\begin{align*}
			\norm{q_i(t)} = \norm{q_{\{i\}}(t) + q_{\{1,2,3,4\}}(t)} \geq \norm{q_{C_2^{(k)}}(t)} - \norm{q_{\{1,2,3,4\}}(t)} + o(1) > R_0.
		\end{align*}
		Here one has to recall that our cluster coordinates are with respect to the center of mass of the subsystem $U_1$ respectively $U_2$, whereas the single coordinates $q_i$ are with respect to the origin in $\R^{d}$. But thus we can estimate 
		\begin{align*}
			\norm{q_i(t)-q_j(t)} \geq \Con{1} \norm{q_i(t)}
		\end{align*}
		and use this for estimating the integral. Additionally we get
		\begin{align*}
			\norm{q_i(t)} &\geq \frac{\norm{q_{C_2^{(1)}}(t)}}{2},
			\intertext{hence}
			\int_{x^{(k)}}^{u_{\text{in}}} \norm{q_i(t)-q_j(t)}^{-(1+\alpha)} \diff t &\leq \operatorname{Const}
			\int_{x^{(k)}}^{u_{\text{in}}} \norm{q_{C_2^{(k)}}(t)}^{-(1+\alpha)}.
			\intertext{Now we can use the propagation estimate in negative time direction and obtain}
			&\leq \frac{\operatorname{Const}}{\norm{p_{C_2^{(k)}}\left(x^{(k)}\right)} \norm{q_{C_2^{(k)}}\left(u_{\text{in}}\right)}^{-\alpha}} \leq \\& \leq \frac{\operatorname{Const} 2^\alpha\left(R_0+\mu_q\right)^{\alpha}}{\norm{p_{C_2^{(k)}}\left(x^{(k)}\right)}}.
			\intertext{For the interval $\left[u_{\text{out}}, s^{(k)}\right]$ we get a similar estimate:}
			\int_{u_{\text{out}}}^{s^{(k)}} \norm{q_i(t)-q_j(t)}^{-(1+\alpha)} \diff t &\leq \frac{\operatorname{Const} 2^\alpha\left(R_0+\mu_q\right)^{\alpha}}{\norm{p_{C_2^{(k)}}\left(x^{(k)}\right)}}.
		\end{align*}
		Finally it remains to consider the time interval $\left[s^{(k)}, x^{(k+1)}\right]$. Here we will use a comparison with $\norm{q_{C_1^{(k)}}}$. We can use Lemma~\ref{lem:subsysDoubleCones} as all norms of $q_i$ are large, and then we can compare them by a constant with $\norm{q_{C_1^{(k)}}}$:
		\begin{align*}
			\int_{s^{(k)}}^{x^{(k+1)}} \norm{q_i(t)-q_j(t)}^{-(1+\alpha)} \diff t &\leq \operatorname{Const} \int_{s^{(k)}}^{x^{(k+1)}} \norm{q_{C_1^{(k)}}(t)}^{-(1+\alpha)} \diff t \leq \\ &\leq \frac{\operatorname{Const}}{\norm{p_{C_1^{(k)}}\left(x^{(k)}\right)} \norm{q_{C_1^{(k)}}\left(s^{(k)}\right)}^{\alpha}}.
		\end{align*}
		Here we have used again a propagation estimate, this time for $q_{C_1^{(k)}}$. As $\norm{q_{C_1^{(k)}}\left(s^{(k)}\right)}^{\alpha}$ is even large, this is a sufficiently small term. \\
		The statement of the proposition now follows from these estimates as we can compare the different momenta by mass-dependent constants, compare Corollary~\ref{cor:SomeComparisons}.
	\end{proof}

	\subsubsection{Conclusions from the bound on the external forces}
	\label{subsubsec:ConclusionsBoundExternalForces}
	
	From Proposition~\ref{prop:BoundIntegratedForcesOnePassage} we can conclude that the conditions on the integrated forces are satisfied for each unbounded subsystem with four particles. Therefore, we can apply \cite[Proposition 4.30, Lemma 4.32]{quaschner2025improbability} to obtain the following result:
	
	\begin{proposition}[{Bounds for the orthogonal components in the interval $\left[x^{(i)}, s^{(i)}\right]$}] \label{prop:BoundsOrthogonalComponents}
		Consider the $i$th passage of an unbounded subsystem with sufficiently large odd $i$. Then we have for an arbitrary time $t \in \left[x^{(i)},s^{(i)}\right]$ with $\L$ being some (global) bound for the total angular momentum:
		\begin{align*}
			\norm{\left(q_{C_1^{(i)}}(t)\right)^{\perp\hat{p}_{C_1^{(i)}}\left(t\right)}} &\leq \frac{\operatorname{Const}(\L)}{\norm{p_{C_1^{(i)}}\left(t\right)}},\\ 
			\norm{\left(p_{C_2^{(i)},C_3^{(i)}}^{I}(t)\right)^{\perp \hat{p}_{C_1^{(i)}}(t)}} &\leq \frac{\operatorname{Const}(\L)}{\norm{q_{C_2^{(i)}, C_3^{(i)}}^{I}(t)}}, \\ 
			\norm{\left(q_{C_2^{(i)},C_3^{(i)}}^{I}(t)\right)^{\perp\hat{p}_{C_1^{(i)}}\left(t\right)}} &\leq \frac{\operatorname{Const}(\L)}{\norm{p_{C_2^{(i)},C_3^{(i)}}^{I}\left(t\right)}}.
		\end{align*}
		For $t = x^{(i)}$ (and even on a time interval up to some time $s_{\min} \geq x^{(i)}$) we also have
		\begin{align*}
			\norm{\left(p_{C_3^{(i)}}(t)\right)^{\perp\hat{q}_{C_3^{(i)}}\left(t\right)}} &\leq \frac{\operatorname{Const}(\L)}{\norm{q_{C_3^{(i)}}\left(t\right)}},\\ 
			\norm{\left(p_{C_1^{(i)},C_2^{(i)}}^{I}(t)\right)^{\perp \hat{q}_{C_3^{(i)}}(t)}} &\leq \frac{\operatorname{Const}(\L)}{\norm{q_{C_1^{(i)}}(t)}}, \\ 
			\norm{\left(q_{C_1^{(i)},C_2^{(i)}}^{I}(t)\right)^{\perp\hat{q}_{C_3^{(i)}}\left(t\right)}} &\leq \frac{\operatorname{Const}(\L)}{\norm{p_{C_1^{(i)}, C_2^{(i)}}^{I}\left(t\right)}}.
		\end{align*}
	\end{proposition}
	
	However, this bound is valid only on a subinterval of the passage. To obtain a set of finite volume for the whole passage, we need to investigate the motion before the time $x^{(i)}$ and after the time $s^{(i)}$.
	
	\subsection{Investigation of the motion on an extended interval}
	\label{subsec:InvestigationMotionExtendedInterval}
	
	By now, we can only control the motion of an unbounded subsystem with four particles on the interval $\left[x^{(i)}, s^{(i)}\right]$. To obtain a set of finite volume, in which the subsystem is at any time, we need to investigate also the time intervals $\left[s^{(i-1)}, x^{(i)}\right]$ and $\left[s^{(i)}, x^{(i+1)}\right]$. We will need some further reference times for this, compare \cite[Definition 4.3]{quaschner2023non}:
	
	\begin{definition}[Some further reference times] \label{def:FurtherReferenceTimes} \quad \\
		We consider the $k$th passage of an unbounded subsystem with four particles with odd index $k$. For future uses we define some further reference times and give a pictorial view of these times for the interesting passage in $\left[s^{(k-1)}, x^{(k+1)}\right]$ in Figure~\ref{fig:SkizzeZeiten0_t0_s}. Then we define the time $\tnez$ as the maximal time before $x^{(k)}$ with
		\begin{align*}
			\norm{q_{C_1^{(k)}, C_2^{(k)}}^{I} \left(\tnez\right)} = \norm{p_{C_1^{(k)}, C_2^{(k)}}^{I} \left(x^{(k)}\right)}^{-\gamma},
		\end{align*}
		where $1<\gamma < \frac{2}{\alpha}$ is a fixed constant (close to $\frac{2}{\alpha}$). \\ Similarly, we define $\tnzd$ as the minimal time after $s^{(k)}$ with
		\begin{align*}
			\norm{q_{C_2^{(k)}, C_3^{(k)}}^{I} \left(\tnzd\right)} = \norm{p_{C_2^{(k)}, C_3^{(k)}}^{I} \left(x^{(k)}\right)}^{-\gamma}.
		\end{align*}
		We also need one intermediate time $u_{ch}^{(k)}$, compare also \cite[Definition 4.24]{quaschner2025improbability}
		\begin{align*}
			\norm{q_{C_1^{(k)}, C_2^{(k)}}^{I} \left(u_{ch}^{(k)}\right)} &= \left(1+\frac{m_{C_3^{(k)}}}{m_{C_1^{(k)}}}\right) \norm{q_{C_2^{(k)}, C_3^{(k)}}^{I} \left(u_{ch}^{(k)}\right)}.
		\end{align*}
		For passages with even index $k$, one has to swap the roles of $C_1^{(k)}$ and $C_3^{(k)}$ in the definition as usual.
	\end{definition}
	
	\begin{figure}[htbp]
		\includegraphics[width=\textwidth]{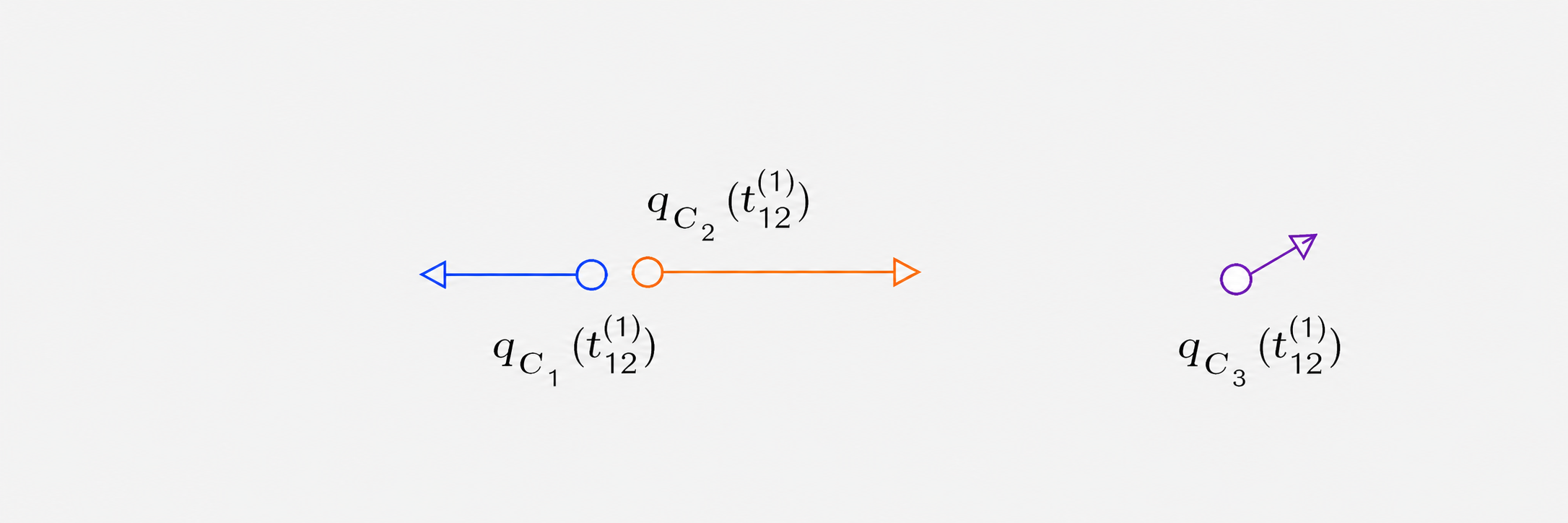}
		\includegraphics[width=\textwidth]{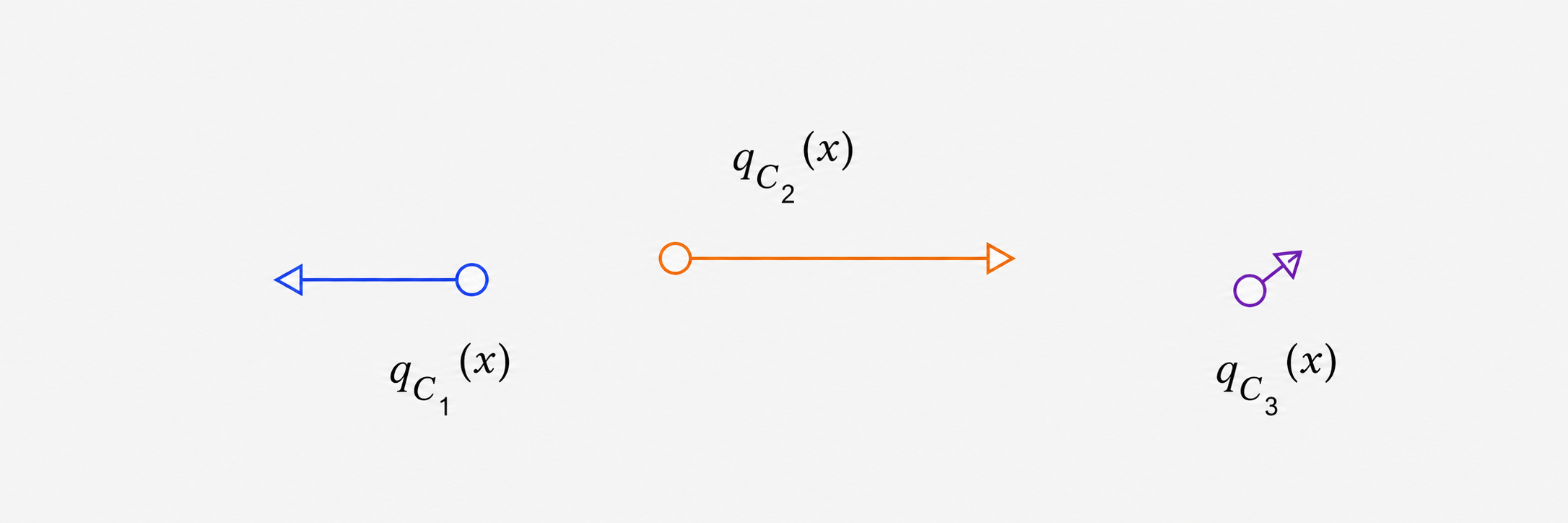}
		\includegraphics[width=\textwidth]{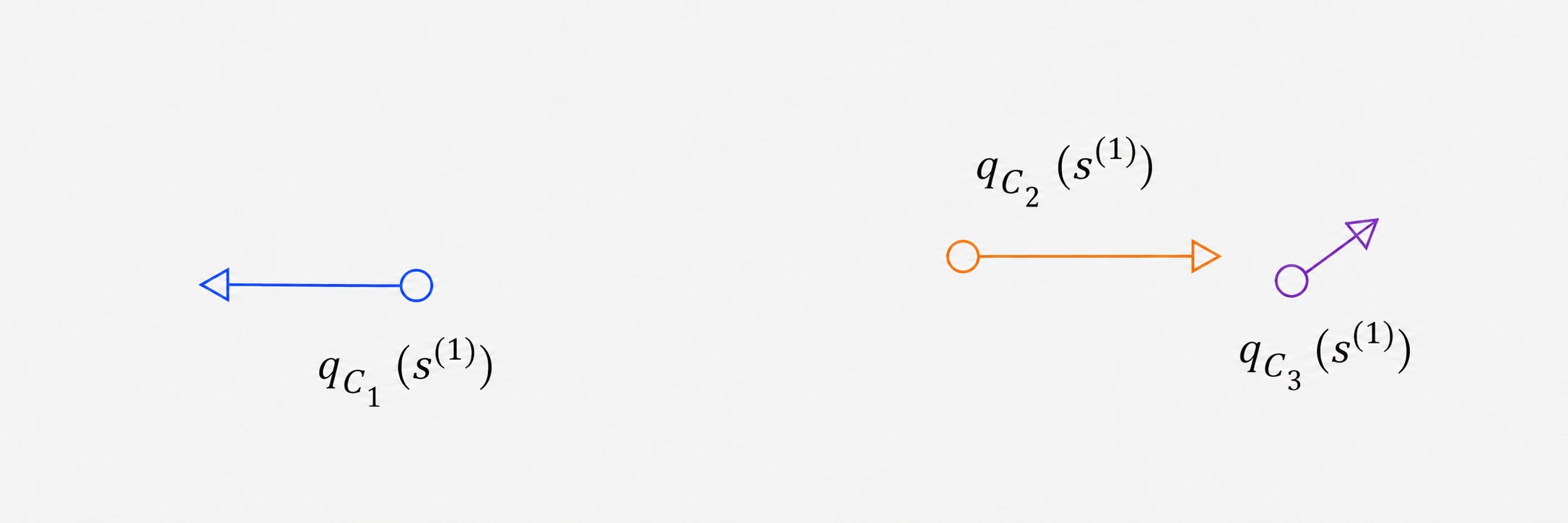}
		\includegraphics[width=\textwidth]{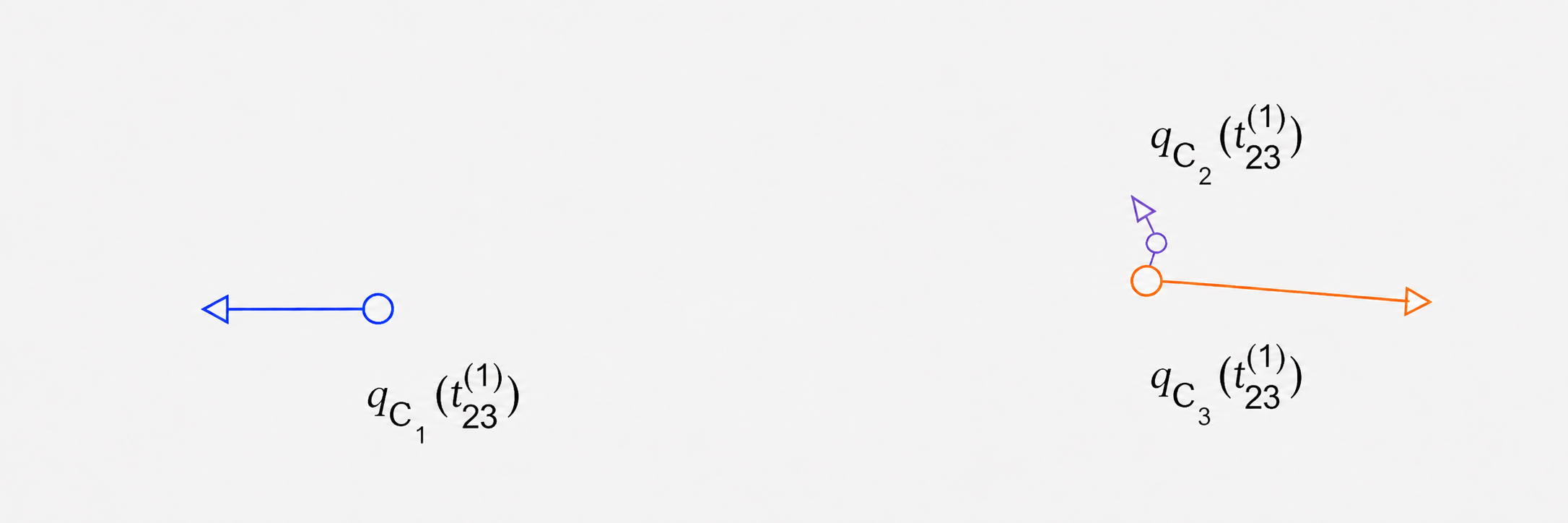}
		\caption[Sketch of the times during one passage: $\tnez$, $x^{(k)}$, $s^{(k)}$, $\tnzd$.]{The small pictures represent the times  $\tnez$, $x^{(k)}$, $s^{(k)}$, $\tnzd$ from the above definition for the first passage. The dots represent the position of the corresponding cluster barycenters and the arrows show the corresponding momenta.}
		\label{fig:SkizzeZeiten0_t0_s}
	\end{figure}

	These reference times are well-defined, as is proven in \cite[Propositions 4.6, 4.8]{quaschner2025improbability} using some propagation estimates. With the help of these propagation estimates, one can control the change of the momenta in a sufficient way. Outside of the interval $\left[\tnez, \tnzd \right]$ there are always three bodies close to each other. This yields a complicated near-collision dynamics of three bodies that is not well-understood by now. However, it is enough that these particles are close together to define suitable bounds for the motion.
	\\
	We will first consider the motion on the interval $\left[\tnez, x^{(k)}\right]$ and $\left[s^{(k)}, \tnzd\right]$. Afterwards, we will give some bounds on the complicated regions $\left[t_{12}^{(k-1)}, \tnez\right]$ and $\left[\tnzd, t_{23}^{(k+1)}\right]$.

	\subsubsection{Refined Investigation of the Motion before Time $x^{(k)}$}
	\label{sec:RigInvSmallTimesVierTeilchen}
	
	To investigate the motion in the interval $\left[\tnez, x^{(1)}\right]$,  we fix $\hat{q}_{C_3^{(k)}}\left(x^{(k)}\right)$ as a common reference direction. This seems plausible, since cluster $C_3$ is far away; hence only small forces act to change the momentum, and from bounds on the angular momentum the directions of position and momentum are almost aligned. The following lemma proves that the direction of the position does not change much:
	\begin{lemma}[{Change of $\hat{q}_{C_3}$ in the interval $\left[\tnez,x^{(k)}\right]$}] \label{lem:ChangeDirqC3tnezx} \quad \\
		For a passage with odd index $k$ we have:
		\begin{align*}
			\int_{\tnez}^{x^{(k)}} \norm{\dot{\hat{q}}_{C_3^{(k)}}(u)} \diff u \leq \frac{\operatorname{Const} \cdot \L}{\norm{q_{C_3^{(k)}}\left(x^{(k)}\right)}^2 \norm{p_{C_1^{(k)},C_2^{(k)}}^{I}\left(x^{(k)}\right)}}.
		\end{align*}
	\end{lemma}
	
	\begin{proof}
		\quad \\
		First, note that that in the interval $\left[\tnez, \left(x^{(k)}\right)\right]$
		\begin{align*}
			\norm{L_{C_3^{(k)}}} \leq \operatorname{Const} \cdot \L
		\end{align*}
		as one can simply integrate the change of the angular momentum with our bounds on the forces.
		Additionally, from the propagation estimate it becomes apparent that
		\begin{align*}
			x^{(k)}- \tnez \leq \frac{\operatorname{Const}}{\norm{p_{C_1^{(k)},C_2^{(k)}}^{I}\left(x^{(k)}\right)}}.
		\end{align*}
		Then all we need is to integrate the change of the direction and get:
		\begin{align*}
			\int_{\tnez}^{x^{(k)}} \norm{\dot{\hat{q}}_{C_3^{(k)}}(u)} \diff u &= 	\int_{\tnez}^{x^{(k)}} \frac{\norm{L_{C_3^{(k)}}(u)}}{\norm{q_{C_3^{(k)}}(u)}^2} \diff u \leq \frac{\operatorname{Const} \cdot \L \cdot \left(x^{(k)}-\tnez\right)}{\norm{q_{C_3^{(k)}}\left(x^{(k)}\right)}^2} \leq \\
			& \leq \frac{\operatorname{Const} \cdot \L}{\norm{q_{C_3^{(k)}}\left(x^{(k)}\right)}^2 \norm{p_{C_1^{(k)},C_2^{(k)}}^{I}\left(x^{(k)}\right)}}.
		\end{align*}
		Here we used that the change of $\norm{q_{C_3^{(k)}}}$ in this interval is small, i.e. we can compare it with the value at time $x^{(k)}$.
	\end{proof}
	
	We can now analyze the component of the motion orthogonal to $\hat{q}_{C_3^{(k)}}\left(x^{(k)}\right)$. Since we are considering only parts of a single passage, we drop the superscript $k$ from now on. We are only interested in the internal coordinates $q_{C_1,C_2}^{I}$ and $p_{C_1,C_2}^{I}$, but one could make analogous statements on $q_{C_1}, q_{C_2}$. First, we recall the equations of motion for these quantities in a notation suppressing error terms using the $\mathcal{O}$-notation:
	\begin{align}
		&\left(\dot{q}_{C_1, C_2}^{I}\right)^{\perp q_{C_3}(x)} = \frac{\left(p_{C_1, C_2}^{I}\right)^{\perp q_{C_3}(x)}}{m_{C_1, C_2}^{I}}, \notag \\
		\begin{split}
			&\left(\dot{p}_{C_1, C_2}^{I}\right)^{\perp q_{C_3}(x)} = \alpha \left(\sum_{\substack{k \in C_2, \\ j \in C_1}} Z_{k,j}\right) \frac{\left(q_{C_1, C_2}^{I}\right)^{\perp q_{C_3}(x)}}{\norm{q_{C_1,C_2}^{I}}^{2+\alpha}} + \\ & + \mathcal{O}\left(\operatorname{Const} \frac{\norm{q_{D}^{I}}}{\norm{q_{C_1,C_2}^{I}}^{2+\alpha}}\right) + \mathcal{O}\left( \operatorname{Const}\right).
		\end{split} \label{eq:Ablpc1c2IiVierTeilchen}
	\end{align}
	The first error term is compensating the influence from the nontrivial cluster, the second error term is compensating the at most constant external forces and the influence of $C_3$. \\
	Next, we will introduce an orthogonal coordinate frame in $\operatorname{span}\{q_{C_3}(x)\}^{\perp}$ and denote the components of a vector in this space with e.g. $\left(p_{C_1, C_2}^{I}\right)^{\perp q_{C_3}(x)}_{i}$ for $i=1, \ldots, d-1$. Now we can prove the following lemma:
	\begin{lemma}[Interval with an optimal bound for the deviation of the orthogonal component of the position $q_{C_1,C_2}^{I}$]
		\label{lem:OptBoundqC1C2OrthArbCharg}\quad \\
		Consider the $k$-th passage with $k$ odd and sufficiently close to the escape time. \\
		Let $\Con{1} \geq 1$ be a sufficiently large constant such that for any $i=1, \ldots, d-1$ we have
		\begin{align*}
			\left|\left(q_{C_1,C_2}^{I}(x)\right)^{\perp \hat{q}_{C_3}(x)}_i\right| \leq \frac{\Con{1}}{\norm{p_{C_1,C_2}^{I}(x)}}.
		\end{align*}
		Then we have for any time $t \in \left[\tnezo, x \right]$:
		\begin{align}
			\left|\left(q_{C_1,C_2}^{I}(t)\right)^{\perp \hat{q}_{C_3}(x)}_i\right| \leq \frac{2 \Con{1}}{\norm{p_{C_1,C_2}^{I}(x)}}. \label{eq:OptBoundqC1C2Orthattimet}
		\end{align}
	\end{lemma}
	
	\begin{proof}
		\quad \\
		Note that at time $\tnezo$ we have $\left|\left(q_{C_1,C_2}^{I}(\tnezo)\right)^{\perp \hat{q}_{C_3}(x)}_i\right| \leq \norm{p_{C_1,C_2}^{I}(x)}^{-\gamma}$, so if there is a time in the interval $\left[\tnezo, x \right]$ where the bound is broken, there is also a time of a maximum $t_{\max}$ and some time $t_{\operatorname{half}} < t_{\max}$ with the following properties:
		\begin{align*}
			\left|\left(q_{C_1,C_2}^{I}(t_{\max})\right)^{\perp \hat{q}_{C_3}(x)}_i\right| &> \frac{2 \Con{1}}{\norm{p_{C_1,C_2}^{I}(x)}}, \\ 
			\left|\left(q_{C_1,C_2}^{I}(t_{\operatorname{half}})\right)^{\perp \hat{q}_{C_3}(x)}_i\right| &= \frac{\left|\left(q_{C_1,C_2}^{I}(t_{\max})\right)^{\perp \hat{q}_{C_3}(x)}_i\right|}{2}, \\
			\left(p_{C_1,C_2}^{I}(t_{\max})\right)^{\perp \hat{q}_{C_3}(x)}_i &=0.
		\end{align*}
		Additionally, we assume that $t_{\operatorname{half}}$ is the largest time satisfying this condition, such that for any $t \in \left[t_{\operatorname{half}}, t_{\max}\right]$ we have
		\begin{align*}
			\frac{\left|\left(q_{C_1,C_2}^{I}(t_{\max})\right)^{\perp \hat{q}_{C_3}(x)}_i\right|}{2} \leq \left|\left(q_{C_1,C_2}^{I}(t)\right)^{\perp \hat{q}_{C_3}(x)}_i\right| \leq 	\left|\left(q_{C_1,C_2}^{I}(t_{\max})\right)^{\perp \hat{q}_{C_3}(x)}_i\right|.
		\end{align*}
		But then we can estimate:
		\begin{align*}
			\frac{\left|\left(q_{C_1,C_2}^{I}(t_{\max})\right)^{\perp \hat{q}_{C_3}(x)}_i\right|}{2} &= \left|\int_{t_{\operatorname{half}}}^{t_{\max}}\frac{\left(p_{C_1,C_2}^{I}(t)\right)^{\perp \hat{q}_{C_3}(x)}_i}{m_{C_1, C_2}^{I}} \diff t\right| = \\
			& = \left|\int_{t_{\operatorname{half}}}^{t_{\max}} \int_{t}^{t_{\max}}\frac{\left(\dot{p}_{C_1,C_2}^{I}(s)\right)^{\perp \hat{q}_{C_3}(x)}_i}{m_{C_1, C_2}^{I}} \diff s \diff t\right| \leq \\
			& \leq \operatorname{Const} \cdot \int_{t_{\operatorname{half}}}^{t_{\max}} \int_{t}^{t_{\max}} \frac{\left|\left(q_{C_1,C_2}^{I}(s)\right)^{\perp \hat{q}_{C_3}(x)}_i\right|}{\norm{q_{C_1,C_2}^{I}(s)}^{2+\alpha}} \diff s \diff t  + \\ & \qquad +  \operatorname{Const} \cdot (t_{\max} - t_{\operatorname{half}})^2. 
		\end{align*}
		Note that we compensated the error term coming from the size of the nontrivial cluster by the first term, as the distance of the nontrivial cluster is always smaller than $\norm{p_{C_1, C_2}^{I}(x)}^{-2/\alpha}$, whereas $\left|\left(q_{C_1,C_2}^{I}(s)\right)^{\perp \hat{q}_{C_3}(x)}_i\right| \geq \frac{1}{\norm{p_{C_1, C_2}^{I}(x)}}$, i.e. is always larger. Additionally, note that $\operatorname{Const} \cdot (t_{\max} - t_{\operatorname{half}})^2 \leq \operatorname{Const} \cdot \norm{p_{C_1, C_2}^{I}(x)}^{-2}$ is way smaller than the left-hand side. We will investigate now the size of the integral and prove that this is also way smaller than the left-hand side. This will yield the desired contradiction and thus prove the claim. For this we use the propagation estimate twice:
		\begin{align*}
			&\int_{t_{\operatorname{half}}}^{t_{\max}} \int_{t}^{t_{\max}} \frac{\left|\left(q_{C_1,C_2}^{I}(s)\right)^{\perp \hat{q}_{C_3}(x)}_i\right|}{\norm{q_{C_1,C_2}^{I}(s)}^{2+\alpha}} \diff s \diff t  \leq \\ & \leq \operatorname{Const} \left|\left(q_{C_1,C_2}^{I}(t_{\max})\right)^{\perp \hat{q}_{C_3}(s)}_i\right| \int_{t_{\operatorname{half}}}^{t_{\max}} \int_{t}^{t_{\max}} \frac{1}{\norm{q_{C_1,C_2}^{I}(s)}^{2+\alpha}} \diff s \diff t \leq \\
			& \leq \operatorname{Const} \left|\left(q_{C_1,C_2}^{I}(t_{\max})\right)^{\perp \hat{q}_{C_3}(x)}_i\right| \int_{t_{\operatorname{half}}}^{t_{\max}}  \frac{1}{\norm{p_{C_1,C_2}^{I}(x)}\norm{q_{C_1,C_2}^{I}(t)}^{1+\alpha}} \diff t \leq \\
			& \leq \operatorname{Const} \left|\left(q_{C_1,C_2}^{I}(t_{\max})\right)^{\perp \hat{q}_{C_3}(x)}_i\right| \frac{1}{\norm{p_{C_1,C_2}^{I}(x)}^2\norm{q_{C_1,C_2}^{I}(t_{\operatorname{half}})}^{\alpha}}.
		\end{align*}
		But as $\norm{q_{C_1,C_2}^{I}(t_{\operatorname{half}})} \geq \frac{1}{\norm{p_{C_1,C_2}^{I}(x)}}$, this is way smaller than $\left|\left(q_{C_1,C_2}^{I}(t_{\max})\right)^{\perp \hat{q}_{C_3}(x)}_i\right|$. So we have the desired contradiction and the claim is proven.
	\end{proof}
	
	Using the bound on $\norm{\left(q_{C_1,C_2}^{I}\right)^{\perp \hat{q}_{C_3}(x)}}$, we can get by integration with a propagation estimate bounds for $\norm{\left(p_{C_1,C_2}^{I}\right)^{\perp \hat{q}_{C_3}(x)}}$. As the change of $\hat{q}_{C_3}$ was bounded in Lemma~\ref{lem:ChangeDirqC3tnezx}, we can allow arbitrary reference times too and have proven the following proposition:
	
	\begin{proposition}[{Deviation from the straight line before time $x^{(k)}$}] \label{prop:DeviationStraightLineTnezx} \quad \\
		Consider the $k$-th passage with $k$ odd and sufficiently close to the escape time. \\
		Then the following bounds are satisfied at arbitrary times $t \in \left[\tnez,x^{(k)}\right]$:
		\begin{align*}
			&\norm{\left(q_{C_1^{(k)},C_2^{(k)}}^{I}(t)\right)^{\perp \hat{q}_{C_3^{(k)}}(t)}} \leq \frac{\operatorname{Const}(\L)}{\norm{p_{C_1^{(k)},C_2^{(k)}}^{I}\left(t\right)}}, \\
			&\norm{\left({p_{C_1^{(k)},C_2^{(k)}}^{I}(t)}\right)^{\perp \hat{q}_{C_3^{(k)}}(t)}} \leq  \max \left\{ \operatorname{BI}(t), \frac{\operatorname{Const} \L}{\norm{q_{C_3^{(k)}}\left(x^{(k)}\right)}} \right\} \quad \text{ with } \\
			\operatorname{BI}(t) &:= \min \left\{\frac{\operatorname{Const}(\L) }{\norm{{{p_{C_1^{(k)},C_2^{(k)}}^{I}}}\left(t\right)} \norm{{q_{C_1^{(k)},C_2^{(k)}}^{I}}(t)}^{\alpha}},\frac{\operatorname{Const} (\L)}{\norm{{{p_{C_1^{(k)},C_2^{(k)}}^{I}}}\left(t\right)}^2 \norm{{q_{C_1^{(k)},C_2^{(k)}}^{I}}(t)}^{\alpha+1}} \right\}.
		\end{align*}
		In the bound for the integration, $\operatorname{BI}(t)$, the first term is smaller for times close to $\tnez$, whereas the second bound is smaller for times closer to $x^{(k)}$.
	\end{proposition}

	\subsubsection{Investigation of the Motion after Time $s^{(k)}$}
	\label{subsec:OrthKompNachSminVierTeilchen}
	
	This section is similar to Section~\ref{sec:RigInvSmallTimesVierTeilchen}. Therefore, we state only the main results, since the proofs are analogous to those in the previous section. For the reference direction we choose $\hat{p}_{C_1}$ in accordance with Proposition~\ref{prop:BoundsOrthogonalComponents}. Due to the small forces, this will not change much:
	
	\begin{lemma}[{Change of $\hat{p}_{C_1}$ in the interval $[s^{(k)}, \tnzd]$}] \quad \\
		Consider the $k$-th passage with $k$ odd and sufficiently close to the escape time. Then we have
		\begin{align*}
			\int_{s^{(k)}}^{\tnzd} \norm{\dot{\hat{p}}_{C_1^{(k)}}(u)} \diff u \leq \frac{\operatorname{Const}}{\norm{p_{C_1^{(k)}}(x^{(k)})}^2}.
		\end{align*}
	\end{lemma}
	
	Then one can consider $\left(q_{C_2^{(k)}, C_3^{(k)}}^{I}\right)^{\perp \hat{p}_{C_1^{(k)}}\left(s^{(k)}\right)}$. We can use the same trick and consider only one direction in an orthonormal frame of $\operatorname{span}\{p_{C_1^{(k)}}\left(s^{(k)}\right)\}^{\perp}$ and deduce that there can be no internal maximum of size larger than $\frac{\operatorname{Const}}{\norm{p_{C_2,C_3}^{I}\left(x^{(k)}\right)}}$ (compare Lemma~\ref{lem:OptBoundqC1C2OrthArbCharg}). Using this bound, we can integrate the change of the orthogonal momentum. Combining all of this, we have proven the following proposition:
	
	\begin{proposition}[{Deviation from the straight line after $s^{(k)}$}] \label{prop:DeviationStraightLineSminTnzd} \quad \\
		Consider the $k$-th passage with $k$ odd and sufficiently close to the escape time. \\
		Then the following bounds are satisfied at arbitrary times $t \in \left[s^{(k)}, \tnzd\right]$:
		\begin{align*}
			&\norm{\left(q_{C_2^{(k)},C_3^{(k)}}^{I}(t)\right)^{\perp \hat{p}_{C_1^{(k)}}(t)}} \leq \frac{\operatorname{Const} \cdot \L}{\norm{p_{C_2^{(k)},C_3^{(k)}}^{I}(t)}}, \\
			& \norm{({p_{C_2^{(k)},C_3^{(k)}}^{I}(t)})^{\perp \hat{p}_{C_1^{(k)}}(t)}} \leq \max \left\{ \operatorname{BI}_2(t), \operatorname{Const}(\L)  \right\} \quad \text{ with } \\
			\operatorname{BI}_2(t) &:= \min\left\{\frac{\operatorname{Const}(\L) }{\norm{p_{C_2^{(k)},C_3^{(k)}}^{I}(t)}^{2} \norm{{q_{C_2^{(k)},C_3^{(k)}}^{I}}(t)}^{\alpha+1}}, \frac{\operatorname{Const}(\L) }{\norm{{{p_{C_2^{(k)},C_3^{(k)}}^{I}}}(t)} \norm{{q_{C_2^{(k)},C_3^{(k)}}^{I}}(t)}^{\alpha}}  \right\}.
		\end{align*}
	\end{proposition}
	
	\subsubsection{Investigation of the Motion before Time $\tnez$ and/or after Time $\tnzd$}
	\label{subsec:InvMottneztnzd}
	
	The motion during a \enquote{near-collision}, i.e. a close encounter of the particles from two clusters, is most difficult to control. However, as we are only looking for a set of finite volume in which the solution has to stay, we can simplify the analysis. The key idea is that although the momenta can be arbitrary large, the distance of the particles has to be small accordingly. More precisely, let us focus on the interval $\left[\tnzd, \tezd\right]$. Then at the times $\tnzd$ and $\tezd$ we have using an energy argument for the nontrivial cluster $D$ respectively the definition of $\tnzd$ and $\tezd$:
	\begin{align*}
		\norm{q_{D^{(k)}}^{I}\left(\tnzd\right)} &\leq \frac{\operatorname{Const}}{\norm{p_{D^{(k)}}^{I}\left(\tnzd\right)}}, \\
		\norm{q_{C_2^{(k)}, C_3^{(k)}}^{I}\left(\tnzd\right)} &\leq \frac{\operatorname{Const}}{\norm{p_{C_2^{(k)}, C_3^{(k)}}^{I}\left(\tnzd\right)}}, \\
		\norm{q_{D^{(k+1)}}^{I}\left(\tezd\right)} &\leq \frac{\operatorname{Const}}{\norm{p_{D^{(k+1)}}^{I}\left(\tezd\right)}}, \\
		\norm{q_{C_2^{(k+1)}, C_3^{(k+1)}}^{I}\left(\tezd\right)} &\leq \frac{\operatorname{Const}}{\norm{p_{C_2^{(k+1)}, C_3^{(k+1)}}^{I}\left(\tezd\right)}}.
	\end{align*}
	
	The problem is, however, that it is not clear whether these bounds will hold during the whole interval, and even worse, a reformation of clusters is possible (so it is not clear for which sets of clusters these estimates should hold). So first of all we have to introduce some new time-dependent cluster decomposition $\C'(t)$ for the particles in $C_2^{(k)} \cup C_3^{(k)}$ in this time interval. For this we use the potential energy and define
	\begin{align*}
		\mathcal{C}'(t) &:= \left\{ D'(t), \left(C_2^{(k)} \cup C_3^{(k)}\right) \setminus D'(t) \right\}
	\end{align*}
	where we define the non-trivial cluster $D'(t)$ as the term with the smallest\footnote{Note that the total potential energy has to be very negative in order to compensate the large positive kinetic energy. Thus in $D'$ we will consider particles $i,j$ which are close together.} contribution to the total energy
	\begin{align*}
		D'(t) &= \{i,j\} \subseteq C_2^{(k)} \cup C_3^{(k)} \text{ if } \frac{Z_{i,j}}{\norm{q_i-q_j}^{\alpha}} \leq \min_{a,b \in C_2^{(k)} \cup C_3^{(k)}, a \neq b} \frac{Z_{a,b}}{\norm{q_a-q_b}^{\alpha}}.
	\end{align*}
	Note that $D'(t)$ is not uniquely defined at times where there are two pairs realizing the minimum. At those times, a change of the cluster decomposition is possible. As a shorthand notation we denote the clusters in $\mathcal{C}'(t)$ as $C_2'(t), C_3'(t)$ choosing the naming compatible with the next passage, i.e.
	\begin{align*}
		C_3'(t) &:= \begin{cases}
			D'(t) \quad &\text{ if } D'(t) \cap C_3^{(k+1)} \neq \emptyset, \\
			\left(C_2^{(k)} \cup C_3^{(k)}\right) \setminus D'(t) \quad &\text{ else.}
		\end{cases}
	\end{align*}
	The idea to control the motion is now that we should have a condition of the form
	\begin{align*}
		\norm{q_{C_2'(t), C_3'(t)}^{I}(t)} \leq \frac{\operatorname{Const}}{\norm{p_{C_2'(t), C_3'(t)}^{I}(t)}}
	\end{align*}
	during the close encounter and if this condition fails, then we are already past the start of the next passage. This will be made precise in the following proposition: 
	
	\begin{proposition}[Bounds in the phase of a quasi-collision] \label{prop:BoundsPhaseQuasiCollision}\quad \\
		Consider the $k$-th passage with $k$ odd and sufficiently close to the escape time. \\
		Let $\Con{1}$ be a large constant that depends on the usual constants and consider the motion up to the time $\texd \geq \tnzd$ at which we first have
		\begin{align}
			\norm{q_{C_2'\left(\texd\right), C_3'\left(\texd\right)}^{I}\left(\texd\right)} \geq \frac{\Con{1}}{\norm{p_{C_2'\left(\texd\right), C_3'\left(\texd\right)}^{I}\left(\texd\right)}}. \label{eq:DefinitionTex2}
		\end{align}
		Let $\tezd$ be the time at which the next passage \enquote{starts}\footnote{Recall that this means 
			$ \norm{q_{C_2^{(k+1)}, C_3^{(k+1)}}^{I}\left(\tezd\right)} = \norm{p_{C_2^{(k+1)}, C_3^{(k+1)}}^{I}\left(x^{(k+1)}\right)}^{-\gamma}$.}. Assume the following bound on the integrated force terms in the interval $\left[\tnzd, \tezd\right]$ for arbitrary $i \in \{1,2,3,4\}$:
		\begin{align*}
			\int_{\tnzd}^{\tezd} \norm{F_i(t)} \diff t &\leq \frac{\operatorname{Const}}{\norm{p_{C_1^{(k)}}\left(x^{(k)}\right)}},
		\end{align*}
		then the following statements hold:
		\begin{enumerate}[(i)]
			\item \label{bul:BoundsPhaseQuasiColl_EnergyBound} We have for all $t \in \left[\tnzd, \texd\right]$
			\begin{align*}
				\left|H_{\{1,2,3,4\}}^{I}(t)\right| \leq \left|H_{\{1,2,3,4\}}^{I}(\tnzd)\right| + \operatorname{Const}(\delta).
			\end{align*}
			\item \label{bul:BoundsPhaseQuasiColl_NextPassageStarted} $\texd \geq \tezd$ has to hold, i.e. the next passage has already started and we have $C_2'\left(\texd\right) = C_2^{(2)}$, $C_3'\left(\texd\right) = C_3^{(2)}$.
			\item \label{bul:BoundsPhaseQuasiColl_PositionBound} For all times $t \in \left[\tnzd, \texd\right]$
			\begin{align*}
				\norm{q_{C_2'\left(t\right), C_3'\left(t\right)}^{I}\left(t\right)} \leq \frac{\Con{1}}{\norm{p_{C_2'\left(t\right), C_3'\left(t\right)}^{I}\left(t\right)}}.
			\end{align*}
		\end{enumerate}
	\end{proposition}
	
	\begin{remark}
		\quad \\ 
		Note that \eqref{bul:BoundsPhaseQuasiColl_EnergyBound} and \eqref{bul:BoundsPhaseQuasiColl_PositionBound} can be used to define a set of finite volume within which the trajectory will stay during the quasi-collision, whereas \eqref{bul:BoundsPhaseQuasiColl_NextPassageStarted} tells us that with this set we have controlled the motion long enough and now can use the other bounds we derived earlier for the next passage (if the assumptions are still satisfied). \\
		Additionally, one should remark that the integrated force bound in this lemma for the larger time interval is no problem in many of the intended applications. In our case, the perturbing forces are  centered \enquote{close to the origin} so large forces can act on the messenger cluster during its passage but onto the outer particles/clusters the forces are smaller. \\
		In order to prove \eqref{bul:BoundsPhaseQuasiColl_EnergyBound} and \eqref{bul:BoundsPhaseQuasiColl_PositionBound} one only needs some straightforward bounds on the change of the energy and that one is initially below the bound \eqref{eq:DefinitionTex2}. \\
		The difficult part of the statement is thus \eqref{bul:BoundsPhaseQuasiColl_NextPassageStarted}. Here a large problem is that in order to break the bound it is not necessary that $\norm{q_{C_2', C_3'}^{I}}$ is increasing but it is possible to break the bound by a stronger increase of $\norm{p_{C_2', C_3'}^{I}}$. If we have no increase of $\norm{q_{C_2', C_3'}^{I}}$, we cannot apply any kind of propagation estimate in positive time direction. Thus it is not clear, whether or not we reach the level of constant distance of the cluster barycenters at this excursion. And if we do not, we are possibly still before the time $\tezd$. A large part of the proof will be only necessary to rule out this behavior. For this, we need some bounds on the change of the angular momenta (for which we need the additional assumption) that will reduce this question to the influence of $\dot{J}_{C_2', C_3'}$. Then we can see that a decrease cannot lead to a violation of the bound \eqref{eq:DefinitionTex2} and afterwards the propagation estimate argument will succeed. \hfill $\diamond$
	\end{remark}
	
	\begin{proof}
		\quad \\
		First of all we need to notice that at time $\tnzd$ we have
		\begin{align*}
			\norm{q_{C_2'\left(\tnzd\right), C_3'\left(\tnzd\right)}^{I}\left(\tnzd\right)} &= \norm{q_{C_2^{(1)},C_3^{(1)}}^{I}(\tnzd)} = \norm{p_{C_2^{(1)},C_3^{(1)}}^{I}\left(x^{(1)}\right)}^{-\gamma} < \\ & < \frac{\Con{1}}{\norm{p_{C_2^{(1)},C_3^{(1)}}^{I}\left(\tnzd\right)}},
		\end{align*}
		where we can use a propagation estimate for the comparison of the momenta at the times $x^{(k)}$ and $\tnzd$. \\
		Thus, by definition of the time $\texd$, condition \eqref{bul:BoundsPhaseQuasiColl_PositionBound} is satisfied. We will even show that at time $\texd$ we have equality in \eqref{eq:DefinitionTex2} as we cannot get such an increase by a jump of the cluster decomposition. For this, we need the energy considerations of \eqref{bul:BoundsPhaseQuasiColl_EnergyBound}. Note that
		\begin{align*}
			\dot{H}_{\{1,2,3,4\}}^{I} &= \sum_{i=1}^{4} \scalprod{\frac{p_i}{m_i}}{F_i} =  \scalprod{\frac{p_{D'}^{I}}{m_{D'}^{I}}}{F_{D'}^{I}} + \scalprod{\frac{p_{C_2', C_3'}^{I}}{m_{C_2', C_3'}^{I}}}{F_{C_2', C_3'}^{I}} + \\ & \quad + \frac{m_{C_1} + m_{C_2} + m_{C_3}}{m_{C_2} + m_{C_3}} \cdot \scalprod{\frac{p_{C_1}}{m_{C_1}}}{F_{C_1, C_2 \cup C_3}^{I}} +  \scalprod{\frac{p_{\{1,2,3,4\}}}{m_{\{1,2,3,4\}}}}{F_{\{1,2,3,4\}}}
		\end{align*}
		which becomes apparent if one remembers that only the external forces can change the total energy and then we use Jacobi coordinates for the kinetic energy (compare Lemma~\ref{lem:PropJacobiCoordinates}). \\
		But now we can argue with the specific form of the forces in \eqref{eq:ForcesAreCenterlike}: We obtain some cancellation effect for the internal forces, as these force differences mainly depend on the distances of the corresponding barycenters (and the internal distances of the clusters):
		\begin{align*}
			\norm{F_{D'}^{I}} &\leq \operatorname{Const}(\delta) \cdot \norm{q_{D'}^{I}}, \\
			\norm{F_{C_2', C_3'}^{I}} &\leq \operatorname{Const}(\delta) \cdot \left(\norm{q_{D'}^{I}} + \norm{q_{C_2', C_3'}^{I}}\right).
		\end{align*}
		Thus as long as the nontrivial cluster $D'$ provides at least a fraction\footnote{One should note that the numerical value of $\frac{1}{16}$ is somewhat arbitrary. As there are six terms in the potential for four particles, the minimal term must have a contribution of at least $\frac{1}{6}$; we chose a smaller bound as there might be some fluctuations that can thus be compensated too.} of $\frac{1}{16}$ of the total kinetic energy, the nontrivial cluster is sufficiently close together on a scale $K^{-\frac{1}{\alpha}}$. Then the internal contribution to the estimates of $\norm{F_{D'}^{I}}$ and $\norm{F_{C_2', C_3'}^{I}} $ weighted with the corresponding momenta is small:
		\begin{align*}
			\operatorname{Const}(\delta) \cdot \left(\norm{p_{D'}^{I}} + \norm{p_{C_2', C_3'}^{I}}\right) \norm{q_{D'}^{I}} &\leq \operatorname{Const}(\delta) \cdot K^{-\frac{2-\alpha}{2\alpha}}.
		\end{align*}
		For the contribution of $\norm{q_{C_2', C_3'}^{I}}$ to the estimate for $\norm{F_{C_2', C_3'}^{I}}$ we can use that we are before time $\texd$ and thus have \eqref{bul:BoundsPhaseQuasiColl_PositionBound}:
		\begin{align*}
			\operatorname{Const}(\delta) \norm{p_{C_2', C_3'}^{I}} \norm{q_{C_2', C_3'}^{I}} \leq \operatorname{Const}(\delta) \cdot \Con{1}.
		\end{align*}
		Thus the only large term in the derivative of $H_{\{1,2,3,4\}}^{I}$ could be $\scalprod{\frac{p_{C_1^{(k)}}}{m_{C_1^{(k)}}}}{F_{C_1, C_2 \cup C_3}^{I}}$. But as $p_{C_1^{(k)}}$ is almost constant, by an integration of the force term we obtain that this is bounded by $\operatorname{Const}(\delta)$. \\
		Then the proof of the energy bound follows the usual pattern: We take the maximal interval on which the nontrivial cluster provides at least a given fraction of the total kinetic energy $K$ that is smaller than the initial fraction, then we integrate the change of the energy and as this cannot change much, the nontrivial cluster must satisfy this estimate for all times up to $\texd$. Thus \eqref{bul:BoundsPhaseQuasiColl_EnergyBound} is proven. Additionally we see that the nontrivial cluster $D'$ can only change if particles $i,j$ are close together on a scale of the total kinetic energy:
		\begin{align*}
			\norm{q_i-q_j} \leq \operatorname{Const} \cdot K^{-\frac{1}{\alpha}} \ll \frac{\Con{1}}{\norm{p_{\{i,j\}}^{I}}}.
		\end{align*}
		Hence, we cannot violate the bound in \eqref{eq:DefinitionTex2} by a change of the cluster decomposition and thus by continuity we get equality at time $\texd$.\\
		In order to prove \eqref{bul:BoundsPhaseQuasiColl_NextPassageStarted}, we argue by contradiction assuming that $\tezd > \texd$ would hold. Then we can use our bounds for the integrated forces in the time interval $\left[\tnzd, \tezd\right]$. We need this force bound to control the changes of the angular momenta:
		\begin{align*}
			\int_{\tnzd}^{\texd} \norm{\dot{L}(u)} \diff u & = o(1), \\
			\int_{\tnzd}^{\texd} \norm{\dot{L}_{\{1,2,3,4\}}(u)} \diff u & = o(1), \\
			\int_{\tnzd}^{\texd} \norm{\dot{L}_{C_1}(u)} \diff u & = o(1).
		\end{align*}
		We control these changes and can decompose the total angular momentum $L$ using Jacobi coordinates:
		\begin{align*}
			L &= L_{\{1,2,3,4\}} + L_{D'} + L_{C_2', C_3'} + \frac{m_{C_1^{(k)}} + m_{C_2^{(k)}} + m_{C_3^{(k)}}}{m_{C_2^{(k)}} + m_{C_3^{(k)}}} L_{C_1^{(k)}}.
		\end{align*}
		As by our energy considerations we know that $L_{D'}$ is always small, on any interval on which the cluster decomposition does not change, the change of $L_{C_2', C_3'}$ is $o(1)$. But initially (at the first time this cluster decomposition was realized) we already know\footnote{Either this happened at time $\tnzd$ or we can deduce this from the energy considerations at the time of the new cluster formation, as then the distances between two pairs of the three particles from $C_2 \cup C_3$ are small.} that $\norm{L_{C_2', C_3'}} = o(1)$, thus we have on the whole time interval $\left[\tnzd, \texd\right]$:
		\begin{align*}
			\norm{L_{C_2', C_3'}} = o(1).
		\end{align*}
		From this observation we deduce that in order to break the bound \eqref{eq:DefinitionTex2}, we must have
		\begin{align*}
			\left|\scalprod{q_{C_2'\left(\texd\right), C_3'\left(\texd\right)}^{I}\left(\texd\right)}{p_{C_2'\left(\texd\right), C_3'\left(\texd\right)}^{I}\left(\texd\right)}\right| \geq \Con{1} - 1,
		\end{align*}
		as the orthogonal part cannot have a large contribution. \\
		Our goal is now to prove that we can apply a propagation estimate in a suitable interval containing $\texd$ and where there are no changes of the cluster decomposition. Then the change of the momentum is small and if $\norm{q_{C_2'\left(\texd\right), C_3'\left(\texd\right)}^{I}}$ is strictly increasing up to the time $x^{(k+1)}$, we can argue that $\texd > \tezd$ has to hold. \\
		An important observation for the propagation estimate is now that for any $t \in \left[\tnzd, \texd\right]$ we have that
		\begin{align*}
			\norm{q_{C_2'\left(t\right), C_3'\left(t\right)}^{I}\left(t\right)} \leq 1
		\end{align*}
		has to hold. Otherwise there would be a time of a local maximum for $\norm{q_{C_2'\left(t\right), C_3'\left(t\right)}^{I}}$ that is strictly larger than 1. But then due to the at most constant forces, the small remaining time and the impossibility of further cluster changes as long as the clusters are far apart, these clusters could never come close enough to get the acceleration needed to hit cluster $C_1$ again.\\
		From this bound on the position we get by specializing to the time $\texd$:
		\begin{align*}
			\norm{p_{C_2'\left(\texd\right), C_3'\left(\texd\right)}^{I}\left(\texd\right)} &\geq \Con{1} - 1,
		\end{align*}
		hence we can choose this large compared to other terms that only depend on e.g. the mass constants. This minimal bound on the momentum is important for the propagation estimate. \\
		Now we argue in the usual propagation estimate style: Consider the maximal interval $\left[u_1, u_2\right] \subseteq \left[ \tnzd, x^{(k+1)} \right]$ containing $\texd$ and on which we have the following inequalities for any $t \in \left[u_1, u_2\right]$:
		\begin{align}
			\frac{\norm{p_{C_2'\left(\texd\right), C_3'\left(\texd\right)}^{I}(t)}}{\norm{p_{C_2'\left(\texd\right), C_3'\left(\texd\right)}^{I}\left(\texd\right)}}  &\in \left[\frac{1}{2}, 2\right], \label{eq:MomentaBoundPC2'C3'}\\
			\left|\scalprod{q_{C_2'\left(\texd\right), C_3'\left(\texd\right)}^{I}\left(t\right)}{p_{C_2'\left(\texd\right), C_3'\left(\texd\right)}^{I}\left(t\right)}\right| &\geq \frac{\Con{1}}{2}. \label{eq:AbsValueBoundDotJC2'C3'}
		\end{align}
		Note that condition \eqref{eq:AbsValueBoundDotJC2'C3'} implies that
		\begin{align}
			\norm{q_{C_2'\left(\texd\right), C_3'\left(\texd\right)}^{I}\left(t\right)} & \geq  \frac{\Con{1}}{2 \norm{p_{C_2'\left(\texd\right), C_3'\left(\texd\right)}^{I}\left(t\right)}}, \label{eq:LowerBoundqC2'C3'}
		\end{align}
		hence there is no change of the cluster decomposition possible (and therefore we take the cluster decomposition at time $\texd$). As we are only considering times in this interval, in the following we will not write the time argument for the cluster decomposition $C_2', C_3'$ in order to shorten the notation. \\
		As usual for a propagation estimate, the next step consists in considering $\ddot{J}_{C_2', C_3'}^{I}$. By the usual methods we get the following expression for this quantity involving some error terms:
		\begin{align*}
			\ddot{J}_{C_2', C_3'}^{I} &=  \frac{\norm{p_{C_2', C_3'}^{I}}^2}{m_{C_2', C_3'}^{I}} + \frac{\sum\limits_{\substack{i \in C_2', \\ j \in C_3'}} \alpha Z_{i,j}}{\norm{q_{C_2', C_3'}^{I}}^{\alpha}} +  \mathcal{O}\left( \operatorname{Const} \frac{\norm{q_{C_2', C_3'}^{I}}}{\norm{q_{C_1^{(k)}} - q_{C_3'}}^{1+\alpha}}\right)+ \\ & \quad + \mathcal{O}\left(\operatorname{Const} \frac{\norm{q_{D'}^{I}}}{\norm{q_{C_2', C_3'}^{I}}^{1+\alpha}}\right) +  \mathcal{O}\left(\norm{q_{C_2', C_3'}^{I}} \norm{F_{C_2', C_3'}^{I}} \right).
		\end{align*}
		The error terms are always negligible as the momentum is large in this interval and due to the lower bound \eqref{eq:LowerBoundqC2'C3'} for $\norm{q_{C_2', C_3'}^{I}}$ we see that
		\begin{align*}
			\ddot{J}_{C_2', C_3'}^{I} &\geq \frac{\norm{p_{C_2', C_3'}^{I}\left(\texd\right)}^2}{6 m_{C_2', C_3'}^{I}}
		\end{align*}
		holds on the whole interval. Now we will show that the interval $[u_1, u_2]$ is maximal in at least one direction. For this we need to make a case distinction on the sign of $\dot{J}_{C_2', C_3'}^{I}\left(\texd\right)$. One might expect that the only way for  the condition \eqref{eq:DefinitionTex2} to occur for $\texd$ would be a growth of the distance $\norm{q_{C_2', C_3'}^{I}}$ (see the second bullet for this case), but this condition could be broken as well if $\norm{q_{C_2', C_3'}^{I}}$ is decreasing but $\norm{p_{C_2', C_3'}^{I}}$ is increasing more rapidly. That this cannot be the case here will be proven in the first bullet of the case distinction:
		\begin{itemize}
			\item \underline{Case $\dot{J}_{C_2', C_3'}^{I}\left(\texd\right) < 0$:} By the bound on $\ddot{J}_{C_2', C_3'}^{I}$ we see that $\dot{J}_{C_2', C_3'}^{I}$ is monotonically increasing, hence due to the negative sign we see that $\betr{\dot{J}_{C_2', C_3'}^{I}}$ is decreasing in the interval $\left[u_1, \texd \right]$. Therefore, \eqref{eq:AbsValueBoundDotJC2'C3'} is satisfied with strict inequality at time $u_1$ and we use the bound
			\begin{align*}
				J_{C_2', C_3'}^{I}(t) \geq J_{C_2', C_3'}^{I}\left(\texd\right) + \operatorname{Const} \cdot \left(\texd-t\right)^2 \norm{p_{C_2', C_3'}^{I}\left(\texd\right)}^2
			\end{align*}
			for times $t \in \left[u_1, \texd\right]$ to see that the change of $\norm{p_{C_2', C_3'}^{I}}$ is of the smaller order
			\begin{align*}
				\frac{\operatorname{Const}}{\norm{p_{C_2', C_3'}^{I}\left(\texd\right)} \norm{q_{C_2', C_3'}^{I}\left(\texd\right)}^{\alpha}}& \leq \frac{\operatorname{Const} \cdot \norm{p_{C_2', C_3'}^{I}\left(\texd\right)}^{\alpha-1}}{\left(\Con{1}-1\right)^\alpha} \\ & \ll \norm{p_{C_2', C_3'}^{I}\left(\texd\right)}.
			\end{align*}
			Thus by maximality we had $u_1 = \tnzd$, but we know that
			\begin{align*}
				\dot{J}_{C_2, C_3}\left(\tnzd\right) = o(1) \ll \frac{\Con{1}}{2}
			\end{align*}
			which is the desired contradiction.
			\item \underline{Case $\dot{J}_{C_2', C_3'}^{I}\left(\texd\right) > 0$:} Then we get the increase of $\betr{\dot{J}_{C_2', C_3'}^{I}}$ in the interval $\left[\texd, u_2\right]$, thus \eqref{eq:AbsValueBoundDotJC2'C3'} is satisfied with strict inequality at time $u_2$. Then a similar argument as above shows that we have a change of $\norm{p_{C_2', C_3'}^{I}}$ only of smaller order and thus \eqref{eq:MomentaBoundPC2'C3'} is satisfied at this time as well. Maximality then implies $u_2 = x^{(2)}$. But then we have a propagation estimate (compare \cite[Proposition 4.6, p.40ff]{quaschner2025improbability} for more details in a similar context) back to the time $\tezd$, the change of $\norm{p_{C_2', C_3'}^{I}}$ in the interval $\left[\tezd, x^{(2)}\right]$ is only small and $J_{C_2', C_3'}^{I}$ is monotonically increasing in the interval $\left[\min\left\{u_1, \tezd\right\}, x^{(2)}\right]$. Thus, as
			\begin{align*}
				\norm{q_{C_2', C_3'}^{I}\left(\tezd\right)} &= \norm{p_{C_2', C_3'}^{I}\left(x^{(2)}\right)}^{-\gamma} \\ & \ll \frac{\Con{1}}{\norm{p_{C_2', C_3'}^{I}\left(\texd\right)}} = \norm{q_{C_2', C_3'}^{I}\left(\texd\right)},
			\end{align*}
			we deduce that $\tezd < \texd$ has to hold. This contradicts our initial assumption $\tezd > \texd$.
		\end{itemize}
		As we got a contradiction in both cases, the initial assumption $\tezd > \texd$ was false and thus \eqref{bul:BoundsPhaseQuasiColl_NextPassageStarted} is proven.
	\end{proof}
	
	\section{The Poincaré surfaces}
	\label{sec:PoincareSurfaces}
	
	We give a sequence of Poincaré surfaces only for a fixed final cluster decomposition $\C = \{ U_1, \ldots, U_k, B_1, \ldots, B_l\}$ with $k,l \in \N_0$ and $\left|U_i\right|= 4$ for $i=1, \ldots, k$ into unbounded and bounded subsystems, as in Theorem~\ref{thm:ImpResSubsystems}. Without loss of generality we assume that $k>0$, as the case for bounded systems (e.g. having a total collision) was already handled in \cite{duignan2026blowup}. The approach is to define the hypersurfaces as Cartesian products of sets for each subsystem. These sets may depend on parameters (like bounds for the angular momenta), but we can use an exhaustion to get the statement in general.\\
	For $U_1$ we take a sequence of hypersurfaces almost all of which are hit by each perturbed non-collision trajectory. These were constructed in the proof of \cite[Theorem 6.1]{quaschner2025improbability}.\\
	For all other subsystems we take sets of bounded volume in which the corresponding subsystem will stay close to the escape time. For bounded, i.e,. colliding subsystems, these sets have been given in Proposition~\ref{prop:SubsetCollisions}. \\
	What remains to be done is to find such sets for unbounded subsystems with four particles. We can define these sets using the estimates that we derived in the previous sections, then we only need to estimate the volume of these sets. We have to use different sets for the different stages of the passage and the sets will depend on the explicit cluster decomposition during/before a passage; as these are only finitely many, we can take the union over all these decompositions and still obtain a set of finite volume. \\
	We will define these sets now and then calculate the volume. From this, the proof of Theorem~\ref{thm:ImpResSubsystems} follows. This is a short version of \cite[Section 5.3]{quaschner2023non}, where similar calculations can be found with more details. We will only focus on one cluster $U_i$ and for simplicity of the indices we assume that $U_i= \{1,2,3,4\}$. For the definition of the times, compare Definition~\ref{def:FurtherReferenceTimes}.
	
	\begin{definition}[Sets of finite volume for the motion of particles from $\{1,2,3,4\}$] \quad \\
		Fix a cluster decomposition $\{C_1, C_2, C_3\}$ of $\{1,2,3,4\}$ with $D \in \{C_1, C_2, C_3\}$ being the nontrivial cluster. We define the following parametric sets:
		\allowdisplaybreaks
		\begin{align}
			\begin{split}
				&A_{[t_{12}^{(k)}, x_1^{(k)}], \delta, \mathcal{L}, \mu_q, \mu_p} := 
				\\ & \bigg\{ (q_D^{I}, q_{C_1,C_2}^{I}, q_{C_3}, q_{\{1,2,3,4\}}, p_D^{I}, p_{C_1,C_2}^{I}, p_{C_3}, p_{\{1,2,3,4\}}) \in \R^{8d} \mid   \\ & \quad \norm{q_{D}^{I}} \leq \norm{q_{C_1,C_2}^{I}}, \norm{p}^2 \leq \frac{\operatorname{Const}}{\norm{q_D^{I}}^{\alpha}}, \\ & \quad H_{\{1,2,3,4\}}(p,q) \in [-\operatorname{Const}\cdot \log(\norm{p_{C_3}}),\operatorname{Const}\cdot \log(\norm{p_{C_3}})], \\ & \quad \norm{p_{C_3}} \geq 1, \quad -V(q) \leq \operatorname{Const} \norm{q_D^{I}}^{-\alpha}, \norm{q_{\{1,2,3,4\}}} \leq \mu_q, \norm{p_{\{1,2,3,4\}}} \leq \mu_p, \\ & \quad \norm{\left({p_{C_1,C_2}^{I}}\right)^{\perp \hat{q}_{C_3}}} \leq \operatorname{Const}\cdot \mathcal{L} \cdot \max \left\{\frac{1}{\norm{{{p_{C_1,C_2}^{I}}}} \norm{{q_{C_1,C_2}^{I}}}^{\alpha}}, \frac{1}{\norm{q_{C_3}}} \right\}, \\ & \quad \norm{q_{C_1, C_2}^{I}} \leq 1, \norm{\left(q_{C_1,C_2}^{I}\right)^{\perp \hat{q}_{C_3}}} \leq \frac{\operatorname{Const} \cdot \mathcal{L}}{\norm{p_{C_1, C_2}^{I}}} ,  \\ &\quad \norm{p_{C_3}^{\perp \hat{q}_{C_3}}} \leq \frac{\operatorname{Const} \cdot \mathcal{L}}{\norm{q_{C_3}}}, \norm{p_{C_1, C_2}^{I}} \geq 1, 1 \leq \norm{q_{C_3}} \leq \operatorname{Const} \norm{p_{C_1, C_2}^{I}} \bigg\}.
			\end{split}
			\label{def:Aunezx}
		\end{align}
		\begin{align}
			\begin{split}
				& A_{[x^{(k)}, u_{ch}^{(k)}], \delta, \mathcal{L}, \mu_q, \mu_p} := \\
				&\bigg\{ (q_D^{I}, q_{C_1,C_2}^{I}, q_{C_3}, q_{\{1,2,3,4\}}, p_D^{I}, p_{C_1,C_2}^{I}, p_{C_3}, p_{\{1,2,3,4\}}) \in \R^{8d} \mid   \\ & \quad 1 \leq \norm{q_{C_1,C_2}^{I}} \leq \operatorname{Const} \norm{q_{C_3}}, \norm{p}^2 \leq \frac{\operatorname{Const}}{\norm{q_D^{I}}^{\alpha}}, \\ & \quad H_{\{1,2,3,4\}}(p,q) \in [-\operatorname{Const}\cdot \log(\norm{p_{C_3}}),\operatorname{Const}\cdot \log(\norm{p_{C_3}})], \\& \quad \norm{p_{C_3}} \geq 1,-V(q) \leq \operatorname{Const} \norm{q_D^{I}}^{-\alpha}, \norm{q_{\{1,2,3,4\}}} \leq \mu_q, \norm{p_{\{1,2,3,4\}}} \leq \mu_p, \\ & \quad \norm{\left(q_{C_1, C_2}^{I}\right)^{\perp \hat{q}_{C_3}}} \leq \frac{\operatorname{Const} \cdot \mathcal{L}}{\norm{p_{C_1, C_2}^{I}}}, \norm{\left({p_{C_1,C_2}^{I}}\right)^{\perp \hat{q}_{C_3}}} \leq \frac{\operatorname{Const}\cdot \mathcal{L}}{\norm{q_{C_3}}}, \\ & \quad \norm{p_{C_3}^{\perp \hat{q}_{C_3}}} \leq \frac{\operatorname{Const} \cdot \mathcal{L}}{\norm{q_{C_3}}}, \norm{p_{C_1, C_2}^{I}} \geq 1, 1\leq \norm{q_{C_3}} \leq \operatorname{Const} \norm{p_{C_1, C_2}^{I}} \bigg\}.
			\end{split} 
			\label{def:Axsmin}
		\end{align}
		\begin{align}
			\begin{split}
				& A_{[u_{ch}^{(k)}, s^{(k)}], \delta, \mathcal{L}, \mu_q, \mu_p} := \\ & \bigg\{ (q_D^{I}, q_{C_2,C_3}^{I}, q_{C_1}, q_{\{1,2,3,4\}}, p_D^{I}, p_{C_2,C_3}^{I}, p_{C_1}, p_{\{1,2,3,4\}}) \in \R^{8d} \mid   \\ & \quad 1 \leq \norm{q_{C_2,C_3}^{I}} \leq \operatorname{Const} \norm{q_{C_1}}, \norm{p}^2 \leq \frac{\operatorname{Const}}{\norm{q_D^{I}}^{\alpha}}, \\ & \quad H_{\{1,2,3,4\}}(p,q) \in [-\operatorname{Const}\cdot \log(\norm{p_{C_1}}),\operatorname{Const}\cdot \log(\norm{p_{C_1}})], \\& \quad \norm{p_{C_1}} \geq 1,E-V(q) \leq \operatorname{Const} \norm{q_D^{I}}^{-\alpha}, \norm{q_{\{1,2,3,4\}}} \leq \mu_q, \norm{p_{\{1,2,3,4\}}} \leq \mu_p, \\ & \quad \norm{\left(q_{C_2, C_3}^{I}\right)^{\perp \hat{p}_{C_1}}} \leq \frac{\operatorname{Const} \cdot \mathcal{L}}{\norm{p_{C_2, C_3}^{I}}}, \norm{\left({p_{C_2,C_3}^{I}}\right)^{\perp \hat{p}_{C_1}}} \leq \frac{\operatorname{Const}\cdot \mathcal{L}}{\norm{q_{C_2, C_3}^{I}}}, \\ & \quad \norm{q_{C_1}^{\perp \hat{p}_{C_1}}} \leq \frac{\operatorname{Const} \cdot \mathcal{L}}{\norm{p_{C_1}}}, \norm{p_{C_2, C_3}^{I}} \geq 1, 1\leq \norm{q_{C_1}} \leq \operatorname{Const} \norm{p_{C_1}} \bigg\}.
			\end{split} 
			\label{def:Asmins1}
		\end{align}
		\begin{align}
			\begin{split}
				& A_{[s^{(k)}, t_{23}^{(k)}], \delta, \mathcal{L}, \mu_q, \mu_p} := \\ & \bigg\{ (q_D^{I}, q_{C_2,C_3}^{I}, q_{C_1}, q_{\{1,2,3,4\}}, p_D^{I}, p_{C_2,C_3}^{I}, p_{C_1}, p_{\{1,2,3,4\}}) \in \R^{8d} \mid   \\ & \quad \norm{p}^2 \leq \frac{\operatorname{Const}}{\norm{q_D^{I}}^{\alpha}}, H_{\{1,2,3,4\}}(p,q) \in [-\operatorname{Const}\cdot \log(\norm{p_{C_1}}),\operatorname{Const}\cdot \log(\norm{p_{C_1}})], \\& \quad -V(q) \leq \operatorname{Const} \norm{q_D^{I}}^{-\alpha}, \norm{q_{\{1,2,3,4\}}} \leq \mu_q, \norm{p_{\{1,2,3,4\}}} \leq \mu_p, \norm{q_{C_2, C_3}^{I}} \leq 1, \\ & \quad  \norm{\left({p_{C_2,C_3}^{I}}\right)^{\perp \hat{p}_{C_1}}} \leq \operatorname{Const}\cdot \mathcal{L} \cdot \max \left\{\frac{1}{\norm{{{p_{C_2,C_3}^{I}}}} \norm{{q_{C_2,C_3}^{I}}}^{\alpha}}, 1 \right\}, \\ & \quad \norm{\left(q_{C_2,C_3}^{I}\right)^{\perp \hat{p}_{C_1}}} \leq \frac{\operatorname{Const} \cdot \mathcal{L}}{\norm{p_{C_2, C_3}^{I}}}, \norm{q_{C_1}^{\perp \hat{p}_{C_1}}} \leq \frac{\operatorname{Const} \cdot \mathcal{L}}{\norm{p_{C_1}}}, \\ & \quad \norm{p_{C_1}} \geq 1, \norm{p_{C_2, C_3}^{I}} \geq 1, 1 \leq \norm{q_{C_1}} \leq \operatorname{Const} \norm{p_{C_1}},  \norm{q_{D}^{I}} \leq \norm{q_{C_2,C_3}^{I}}\bigg\}.
			\end{split}
			\label{def:As1unzd}
		\end{align}
		\begin{align}
			\begin{split}
				& A_{[t_{23}^{(k)}, t_{23}^{(k+1)}], \delta, \mathcal{L}, \mu_q, \mu_p} := \\ & \bigg\{ (q_D^{I}, q_{C_2,C_3}^{I}, q_{C_1}, q_{\{1,2,3,4\}}, p_D^{I}, p_{C_2,C_3}^{I}, p_{C_1}, p_{\{1,2,3,4\}}) \in \R^{8d} \mid   \\ & \quad \norm{p}^2 \leq \frac{\operatorname{Const}}{\norm{q_D^{I}}^{\alpha}}, H_{\{1,2,3,4\}}(p,q) \in [-\operatorname{Const}\cdot \log(\norm{p_{C_1}}),\operatorname{Const}\cdot \log(\norm{p_{C_1}})], \\ & \quad -V(q) \leq \operatorname{Const} \norm{q_D^{I}}^{-\alpha}, \norm{q_{C_2,C_3}^{I}} \leq 1, \norm{p_{C_1}} \geq 1, \norm{q_{\{1,2,3,4\}}} \leq \mu_q, \norm{p_{\{1,2,3,4\}}} \leq \mu_p, \\ & \quad \norm{p_{C_2,C_3}^{I}} \leq \frac{\operatorname{Const}}{\norm{q_{C_2,C_3}^{I}}}, \norm{q_{C_1}^{\perp \hat{p}_{C_1}}} \leq \frac{\operatorname{Const} \cdot \mathcal{L}}{\norm{p_{C_1}}}, 1\leq \norm{q_{C_1}} \leq \operatorname{Const} \norm{p_{C_1}} \bigg\}.
			\end{split} 
			\label{def:AtnzdtezdV2}
		\end{align}
		The set $A_{\delta, \mathcal{L}, \mu_q, \mu_p}$ is defined as the finite union of all these sets.
	\end{definition}
	
	\begin{remark} \quad
		\begin{enumerate}[(i)]
			\item The sets from the above definition are labeled with intervals from one total passage starting at time $t^{(k)}_{23}$ and ending at the start of the next passage at time $t^{(k+1)}_{23}$.
			\item Most of these conditions can be inferred from the above analysis and some standard energy considerations.
			\item All conditions on the orthogonal components of some vectors have been inferred above, compare Propositions~\ref{prop:BoundsOrthogonalComponents},\ref{prop:DeviationStraightLineTnezx},\ref{prop:DeviationStraightLineSminTnzd}.
			\item Note that the energy of the subsystem need not be fixed, but it can grow at most by an additive constant. This will be compensated by the logarithm of the momentum, as the external momenta are at least exponentially lower bounded in the number of passages.
			\item Note that the intermediate time $u_{eq}$ from Definition~\ref{def:FurtherReferenceTimes} has been included. At this time, we switch from the consideration of the internal coordinates $q_{C_1, C_2}^{I}$ to $q_{C_2, C_3}^{I}$.
		\end{enumerate}
	\end{remark}
	
	\begin{proposition}[Finite volume of the set $A_{\delta, \mathcal{L}, \mu_q, \mu_p}$] \quad \\
		The volume of $A_{\delta, \mathcal{L}, \mu_q, \mu_p}$ is finite.
	\end{proposition}
	
	\begin{proof}
		\quad \\
		We prove that the volume of each set whose union $A_{\delta, \mathcal{L}, \mu_q, \mu_p}$ is, is finite. The strategy of the proof is always similar, so we carry out the calculation only for the \enquote{simple} case of $A_{[x^{(k)}, u_{ch}], \delta, \mathcal{L}, \mu_q, \mu_p}$. We will use Fubini's theorem and integrate over the variables in a suitable order.
		\begin{itemize}
			\item The first variable over which we integrate is $p_{D}^{I}$. Note that only the norm of this variable matters, so we can go to polar coordinates. As all other variables are fixed, we can use the energy condition and have to integrate this variable over a hollow $d$-dimensional sphere. The diameter of the shell is roughly $\operatorname{Const} \log(\norm{p_{C_3}})$ and the radius of the sphere is approximately $\norm{q_{C_1, C_2}^{I}}^{-\frac{\alpha}{2}}$, which yields a volume of roughly $\operatorname{Const} \log(\norm{p_{C_3}}) \norm{q_D^{I}}^{-\frac{\alpha}{2} (d-2)}$.
			\item Next, we can integrate the center of mass coordinates, which only yields a multiplicative constant.
			\item Now we go to polar coordinates for all quantities, as the conditions only depend on the norms of the coordinates (we drop out the energy condition now, which will just increase the integral). If we impose a restriction on the orthogonal component of one variable, the corresponding integration region is not a sphere but rather a cylinder with the height as the norm of the variable and for the other directions we get the bound on the orthogonal component (or the norm of the quantity, whatever is smaller). We do all of these estimates, write $r_x$ for the norm of $q_x$ and $R_x$ for the norm of $p_x$ and obtain the following estimate:
		\end{itemize}
		\begin{align*}
			&\lambda^{8d}(A_{[x^{(1)}, s_{\min}], \delta, \kappa_L, \mu_q, \mu_p}) \leq \\
			& \leq \operatorname{Const} \int_{0}^{\operatorname{Const}} \diff r_{D}^{I} \int_{1}^{\infty} \diff r_{C_1, C_2} \int_{1}^{\infty} \diff r_{C_3} \int_{1}^{\infty} \diff R_{C_1, C_2} \int_{1}^{\infty} \diff R_{C_3} \\ & \quad \ind_{\left\{r_{C_1,C_2} \leq \operatorname{Const} r_{C_3}, R_{C_1, C_2} \leq \operatorname{Const} (r_{D}^{I})^{-\frac{\alpha}{2}}, R_{C_3} \leq \operatorname{Const} (r_{D}^{I})^{-\frac{\alpha}{2}}, 1 \leq r_{C_3} \leq \operatorname{Const} R_{C_1, C_2} \right\}} \\ & (r_D^{I})^{-(d-2) \frac{\alpha}{2}} \min \left\{ R_{C_1, C_2}, \frac{1}{r_{C_3}} \right\}^{d-1} \min\left\{r_{C_1,C_2}, \frac{1}{R_{C_1, C_2}} \right\}^{d-1} \min\left\{ R_{C_3}, \frac{1}{r_{C_3}} \right\}^{d-1} \\ & \quad \log(R_{C_3}) (r_D^{I})^{d-1} r_{C_3}^{d-1}.
		\end{align*}
		Now we resolve the minima, using that the fractions are always smaller than $1$, whereas the other variables are larger than $1$. We substitute $r_D := r_D^{I}$ to avoid double superscripts.
			\begin{align*}
			&\leq \operatorname{Const} \int_{0}^{\operatorname{Const}} \diff r_{D} \int_{1}^{\infty} \diff r_{C_1, C_2} \int_{1}^{\infty} \diff r_{C_3} \int_{1}^{\infty} \diff R_{C_1, C_2} \int_{1}^{\infty} \diff R_{C_3} \\ & \quad \ind_{\left\{ r_{C_1,C_2} \leq \operatorname{Const} r_{C_3}, R_{C_1, C_2} \leq \operatorname{Const} r_{D}^{-\frac{\alpha}{2}}, R_{C_3} \leq \operatorname{Const} r_{D}^{-\frac{\alpha}{2}}, r_{C_3} \leq \operatorname{Const} R_{C_1, C_2} \right\}} \\ & \quad r_D^{-(d-2) \frac{\alpha}{2} + (d-1)} \frac{1}{r_{C_3}^{d-1}} \frac{1}{R_{C_1, C_2}^{d-1}} \log(R_{C_3}).
		\end{align*}
		Now we integrate all variables except $r_D$, using the bounds for the domain of integration from the indicator function:
		\begin{align*}
			\lambda^{8d}(A_{[x^{(1)}, s_{\min}], \delta, \kappa_L, \mu_q, \mu_p, E}) \leq &\operatorname{Const} \int_{0}^{\operatorname{Const}} \diff r_{D} \; E   \kappa_L^{3(d-1)}  r_D^{(d-1) \left(1-\frac{\alpha}{2} \right) - \frac{\alpha}{2}} \log(r_D)^2 < \infty,
		\end{align*}
		as the exponent is strictly larger than $-1$ and the logarithm can be compensated.
	\end{proof}
	
	\newpage
	\bibliographystyle{plain}
	\bibliography{Bibliography}

\end{document}